\documentclass{llncs}

\usepackage{stmaryrd}
\usepackage{amssymb}
\usepackage{empheq}
\usepackage{verbatim}
\usepackage{amsmath,amsfonts}

\newcommand{\BS}{\widehat}

\newcommand{\cQ}{\overline{Q}}
\newcommand{\cE}{\overline{E}}

\newcommand{\rE}{\widetilde{E}}

\newcommand{\hE}{\widehat{E}}

\newcommand{\cTt}{\overline{\cal T}}
\newcommand{\hTt}{\widehat{\cal T}}
\newcommand{\rTt}{\widetilde{\cal T}}

\newcommand{\rGg}{\widetilde{\Gamma}}
\newcommand{\hGg}{\widehat{\Gamma}}
\newcommand{\cGg}{\overline{\Gamma}}

\newcommand{\hSigmaMax}{\widehat{\Sigma}_\mMAX}
\newcommand{\hSigmaMin}{\widehat{\Sigma}_\mMIN}
\newcommand{\cSigmaMax}{\overline{\Sigma}_\mMAX}
\newcommand{\cSigmaMin}{\overline{\Sigma}_\mMIN}
\newcommand{\rSigmaMax}{\widetilde{\Sigma}_\mMAX}
\newcommand{\rSigmaMin}{\widetilde{\Sigma}_\mMIN}
\newcommand{\SigmaMax}{\Sigma_\mMAX}
\newcommand{\SigmaMin}{\Sigma_\mMIN}
\newcommand{\simpleMax}{\Xi_\mMAX}
\newcommand{\simpleMin}{\Xi_\mMIN}

\newcommand{\RegA}{\textsf{RA}}
\newcommand{\CorA}{\textsf{CP}}

\newcommand{\set}[1]{\{ #1 \}}
\newcommand{\Set}[1]{\big\{ #1 \big\}}
\newcommand{\eset}[1]{\{ \: #1 \: \}}

\newcommand{\seq}[1]{\langle #1 \rangle}
\newcommand{\Rplus}{{\mathbb R}_{\oplus}} 
\newcommand{\Npos}{\mathbb N_{+}}
\newcommand{\Nat}{\mathbb N}

\newcommand{\Real}{\mathbb R}
\newcommand{\Int}{\mathbb{Z}}
 
\newcommand{\aA}{\mathbb{A}}

\newcommand{\Aa}{{\cal A}}

\newcommand{\Gg}{{\cal G}}

\newcommand{\Mm}{{\cal M}}

\newcommand{\Rr}{{\cal R}}

\newcommand{\Tt}{{\cal T}}
\newcommand{\Vv}{{\cal V}}

\newcommand{\CC}{\textsf{SCC}}
\newcommand{\RESET}{\textsf{reset}}
\newcommand{\FRUNS}{\text{Runs}_{\text{fin}}}
\newcommand{\RUNS}{\text{Runs}}
\newcommand{\LAST}{\textsf{last}}

\newcommand{\TYPES}{\text{Types}}
\newcommand{\FTYPES}{\text{Types}_{\text{fin}}}

\newcommand{\CLOSOP}{\textsf{clos}}

\newcommand{\LENGTH}{\textsf{length}}
\newcommand{\VAL}{\textsf{val}}
\newcommand{\UVAL}{\overline{\textsf{val}}}
\newcommand{\LVAL}{\underline{\textsf{val}}}
\newcommand{\SUCC}{\textsf{succ}}
\newcommand{\RUN}{\textsf{run}}
\newcommand{\FLOOR}[1]{\lfloor #1 \rfloor}
\newcommand{\TIME}{\textsf{time}}

\newcommand{\THIN}{\text{Thin}}
\newcommand{\THICK}{\text{Thick}}
\newcommand{\mMIN}{\text{Min}}
\newcommand{\mMAX}{\text{Max}}

\newcommand{\NATS}[1]{\llparenthesis #1 \rrparenthesis_\Nat}
\newcommand{\REALS}[1]{\llparenthesis #1 \rrparenthesis_\Real}

\newcommand{\TYPE}[1]{\llbracket #1 \rrbracket_\Rr}

\newcommand{\FRAC}[1]{\lbag #1 \rbag}

\newcommand{\BOUNDOP}{\textsf{bd}}

\newcommand{\pQ}{\Omega}

\newcommand{\cS}{\overline{S}}
\newcommand{\hS}{\widehat{S}}
\newcommand{\rS}{\widetilde{S}}

\newcommand{\PRERUNS}{\textsf{PreRuns}}
\newcommand{\FPRERUNS}{\textsf{PreRuns}_{\text{fin}}}

\newcommand{\pSigmaMax}{\Sigma^{\text{pre}}_\mMAX}
\newcommand{\pSigmaMin}{\Sigma^{\text{pre}}_\mMIN}

\newcommand{\AVERAGETIME}{\Aa}

\newcommand{\cchi}{X}
\newcommand{\mmu}{M}
\newcommand{\tpq}[2]{#1^{\downarrow #2}}
\newcommand{\tpqs}[3]{#1^{(#2, #3)}}

\newcommand{\EcMin}[2]{\rSigmaMin^{(#1, #2)}}
\newcommand{\EcMax}[2]{\rSigmaMax^{(#1, #2)}}
\newcommand{\Ff}{\mathcal{F}}

\title{Average-Time Games on Timed Automata}
\date{\today}

\author{Marcin Jurdzi{\'n}ski~\inst{1},
  and Ashutosh Trivedi~\inst{2}}

\institute{Department of Computer Science,University of Warwick, UK \\
\and Computing Laboratory, University of Oxford, UK.}

\begin{document}

\maketitle

\begin{abstract}
  An average-time game is played on the infinite graph of
  configurations of a finite timed automaton.
  The two players, Min and Max, construct an infinite run of the
  automaton by taking turns to perform a timed transition.
  Player Min wants to minimise the average time per transition and
  player Max wants to maximise it.
  A solution of average-time games is presented using a reduction to
  average-price game on a finite graph.
  A direct consequence is an elementary proof of determinacy for 
  average-time games.
  This complements our results for reachability-time games and 
  partially solves a problem posed by Bouyer et al., to design an
  algorithm for solving average-price games on priced timed
  automata. 
  The paper also establishes the exact computational complexity of
  solving average-time games: the problem is EXPTIME-complete for
  timed automata with at least two clocks.
\end{abstract}

\section{Introduction}

\emph{Real-time open systems} are computational systems that interact
with environment and whose correctness depends critically on the time
at which they perform some of their actions. 
The problem of design and verification of such systems can be
formulated as \emph{two-player zero-sum games}. 
A heart pacemaker is an example of a real-time open system as it
interacts with the environment (heart, body movements, and breathing)
and its correctness depends critically on the time at which it
performs some of its actions (sending pace signals to the heart in
real time).  
Other examples of safety-critical real-time open systems include
nuclear reactor protective systems, industrial process controllers, 
aircraft-landing scheduling systems, satellite-launching systems,
etc. 
Designing correct real-time systems is of paramount importance.  
Timed automata~\cite{AD94} are a popular and well-established
formalism for modelling real-time systems, and games on timed automata 
can be used to model real-time open systems.
In this paper, we introduce \emph{average-time games} which
model the interaction between the real-time open system and the
environment; and we are interested in finding a strategy of the
system which results in minimum average-time per transition,
assuming adversarial environment.

\noindent{\bf Related Work.}
Games with quantitative payoffs can be studied as a model for
optimal-controller synthesis~\cite{AM99,ABM04,BCFL04}.  
Among various quantitative payoffs the average-price
payoff~\cite{Gil57,FV97} is the most well-studied in game theory, 
Markov decision processes, and planning literature~\cite{FV97,Put94},
and it has numerous appealing interpretations in applications.  
Most algorithms for solving Markov decision processes~\cite{Put94} or
games with average-price payoff work for finite graphs
only~\cite{ZP96,FV97}. 
Asarin and Maler~\cite{AM99} presented the first algorithm for games
on timed automata (timed games) with a quantitative payoff:
reachability-time payoff.  
Their work was later generalised by Alur et al.~\cite{ABM04} and
Bouyer et al.~\cite{BCFL04} to give partial decidability results for
reachability-price games on linearly-priced timed automata.
The exact computational complexity of deciding the value in timed
games with reachability-time payoff was shown to be EXPTIME
in~\cite{JT07,BHPR07}.     
Bouyer et al.~\cite{BBL04} also studied the more difficult
average-price payoffs, but only in the context of scheduling, which in  
game-theoretic terminology corresponds to 1-player games.
They left open the problem of proving decidability of 2-player 
average-reward games on linearly-priced timed automata.
We have recently extended the results of Bouyer et al.\ to solve 
1-player games on more general concavely-priced timed
automata~\cite{JT08}. 
In this paper we address the important and non-trivial special case of
average-time games (i.e., all locations have unit costs), which was
also left open by Bouyer et al.

\noindent{\bf Our Contributions.}
Average-time games on timed automata are introduced. 
This paper gives an elementary proof of determinacy for these games. 
A new type of region~\cite{AD94} based abstraction---boundary
region graph---is defined, which generalises the corner-point
abstraction of Bouyer et al.~\cite{BBL04}.
Our solution allows computing the value of average-time games for an
arbitrary starting state (i.e., including non-corner states).
Finally, we establish the exact complexity of solving average-time 
games: the problem is EXPTIME-complete for timed automata with at
least two clocks.

\noindent{\bf Organisation of the Paper.} 
In Section~\ref{section:average-payoff-games} we discuss
average-price games (also known as mean-payoff games) on finite
graphs and cite some important results for these games.  
In Section~\ref{section:average-time-games} we introduce average-time
games on timed automata.
In Section~\ref{section:some-useful-abstractions} we introduce some
region-based abstractions of timed automata, including the closed
region graph, and its subgraphs: the boundary region graph, and the
region graph.
While the region graph is semantically equivalent to the corresponding 
timed automaton, the boundary region graph has the property that for 
every starting state, the reachable state space is finite. 
We introduce average-time games on these graphs and in
Section~\ref{section:atg-on-region-graphs} we show that if we
have the solution of the average-time game for any of these graphs,
then we get the solution of the average-time game for the
corresponding timed automaton.
Finally, in Section~\ref{section:complexity} we discuss the
computational complexity of solving average-time games. 

\noindent{\bf Notations.}
We assume that, wherever appropriate, sets $\Int$ of integers, 
$\Nat$ of non-negative integers and $\Real$ of reals contain a maximum
element $\infty$, and we write $\Npos$ for the set of positive
integers and $\Rplus$ for the set of non-negative
reals.   
For $n \in \Nat$, we write $\NATS{n}$ for the set 
$\set{0, 1, \dots, n}$, and $\REALS{n}$ for the set 
$\set{r \in \Real \: : \: 0 \leq r \leq n}$ of non-negative reals 
bounded by~$n$.
For a real number $r \in \Real$, we write $|r|$ for its absolute
value, we write $\FLOOR{r}$ for its integer
part, i.e., the largest integer $n \in \Nat$, such that $n \leq r$,
and we write $\FRAC{r}$ for its fractional part, i.e., we have
$\FRAC{r} = r - \FLOOR{r}$.

\section{Average-Price Games}
\label{section:average-payoff-games}

A (perfect-information) two-player 
\emph{average-price game}~\cite{ZP96,FV97} (also known as mean-payofff
game) $\Gamma = (V, E, V_\mMAX, V_\mMIN, p)$ consists of a
finite directed graph $(V, E)$, a partition $V = V_\mMAX \cup V_\mMIN$
of vertices, and a \emph{price function} $\pi : E \to \Int$.  
A play starts at a vertex $v_0 \in V$.
If $v_0 \in V_p$, for $p \in \eset{\mMAX, \mMIN}$, then player~$p$
chooses a successor of the current vertex $v_0$, i.e., a vertex $v_1$,
such that $(v_0, v_1) \in E$, and $v_1$ becomes the new current
vertex.  
When this happens then we say that player~$p$ has made a move from the
current vertex.
Players keep making moves in this way indefinitely, thus forming
an infinite path $r = (v_0, v_1, v_2, \dots)$ in the game graph.  
The goal of player Min is to minimise 
$\Aa_\mMIN(r) = 
\limsup_{n \to \infty} (1/n)\cdot \sum_{i=1}^{n} \pi(v_{i-1}, v_i)$
and the goal of player Max is to maximise 
$\Aa_\mMAX(r) = 
\liminf_{n \to \infty} (1/n)\cdot \sum_{i=1}^{n} \pi(v_{i-1}, v_i)$.

Strategies for players are defined as usual~\cite{ZP96,FV97}. 
We write $\Sigma_\mMIN$ ($\Sigma_\mMAX$) for the set of strategies of
player~Min (Max) and $\Pi_\mMIN$ ($\Pi_\mMAX$)  for the set of
positional strategies of player~Min (Max).
For strategies $\mu \in \Sigma_\mMIN$ and $\chi \in \Sigma_\mMAX$, and
for an initial vertex $v \in V$, we write $\RUN(v, \mu, \chi)$ for the 
unique path formed if players start in the vertex $v$ and then they
follow strategies $\mu$ and $\chi$, respectively.
For brevity, we write $\Aa_\mMIN(v, \mu, \chi)$ for 
$\Aa_\mMIN(\RUN(v, \mu, \chi))$ and we write $\Aa_\mMAX(v, \mu, \chi)$
for $\Aa_\mMAX(\RUN(v, \mu, \chi))$. 

For a vertex $v \in V$, we define the \emph{upper value} as
\[
\UVAL(v) = \inf_{\mu \in \Sigma_\mMIN} \sup_{\chi \in \Sigma_\mMAX} 
\Aa_\mMIN(v, \mu, \chi),
\]
and the \emph{lower value} as
\[
\LVAL(v) = \sup_{\chi \in \Sigma_\mMAX} \inf_{\mu \in \Sigma_\mMIN}
\Aa_\mMAX(v, \mu, \chi).
\]
Note that the inequality $\LVAL(v) \leq \UVAL(v)$ always holds. 
A game is determined if 
for every~$v \in V$, we have $\LVAL(v) = \UVAL(v)$.  
We then write $\VAL(v)$ for this number and we call it the
\emph{value} of the average-price game at the vertex~$v$.  
 
We say that the strategies $\mu^* \in \Sigma_\mMIN$ and 
$\chi^* \in \Sigma_\mMAX$ are \emph{optimal} for the respective
players, if for every vertex $v \in V$, we have that 
$\sup_{\chi \in \Sigma_\mMAX} \Aa_\mMIN(v, \mu^*, \chi) = \UVAL(v)$  
and 
$\inf_{\mu \in \Sigma_\mMIN} \Aa_\mMIN(v, \mu^*, \chi) = \LVAL(v)$.  
Liggett and Lippman~\cite{LL69} show that all perfect-information
(stochastic) average-price games are positionally determined. 

\begin{theorem}
  \cite{LL69}
  \label{theorem:apg-determined}
  Every average-price game is determined, and optimal positional 
  strategies exist for both players, i.e., for all $v \in V$, we have: 
  \[
  \inf_{\mu \in \Pi_\mMIN} \sup_{\chi \in \Sigma_\mMAX}
  \Aa_\mMIN(v, \mu, \chi) 
  = 
  \sup_{\chi \in \Pi_\mMAX} \inf_{\mu \in \Sigma_\mMIN}
  \Aa_\mMAX(v, \mu, \chi) .
  \]
\end{theorem}

The decision problem for average-price games is in 
NP $\cap$ co-NP;
no polynomial-time algorithm is currently known for the problem.

\section{Average-Time Games}
\label{section:average-time-games}

\subsection{Timed Automata}
Before we present the syntax of the timed automata, we need to
introduce some concepts.
Fix a constant $k \in \Nat$ for the rest of this paper.
Let $C$ be a finite set of \emph{clocks}.
Clocks in timed automata are usually allowed to take arbitrary
non-negative real values. 
For the sake of simplicity and w.l.o.g~\cite{BBBR07}, we restrict
them to be bounded by some constant $k$, i.e., we consider only
\emph{bounded} timed automata models.    
A ($k$-bounded) \emph{clock valuation} is a function 
$\nu : C \to \REALS{k}$;
we write $\Vv$ for the set $[C \to \REALS{k}]$ of clock valuations. 
If~$\nu \in \Vv$ and $t \in \Rplus$ then we write $\nu + t$ for the
clock valuation defined by $(\nu + t)(c) = \nu(c) + t$, for all 
$c \in C$. 
For a set $C' \subseteq C$ of clocks and a clock valuation 
$\nu : C \to \Rplus$, we define $\RESET(\nu, C')(c) = 0$ if 
$c \in C'$, and $\RESET(\nu, C')(c) = \nu(c)$ if $c \not\in C'$.
A \emph{corner} is an integer clock valuation, i.e., $\alpha$ is a
corner if $\alpha(c) \in \NATS{k}$, for every clock $c \in C$.

The set of \emph{clock constraints} over the set of clocks $C$ is the 
set of conjunctions of \emph{simple clock constraints}, which are
constraints of the form $c \bowtie i$ or $c - c' \bowtie i$, where 
$c, c' \in C$, $i \in \NATS{k}$, and 
${\bowtie} \in \eset{<, >, =, \leq, \geq}$. 
There are finitely many simple clock constraints.
For every clock valuation $\nu \in \Vv$, let $\CC(\nu)$ be the set of 
simple clock constraints which hold in~$\nu \in \Vv$.  
A \emph{clock region} is a maximal set $P \subseteq \Vv$, such that for 
all $\nu, \nu' \in P$, $\CC(\nu) = \CC(\nu')$.
In other words, every clock region is an equivalence class of the
indistinguishability-by-clock-constraints relation, and vice versa.  
Note that $\nu$ and~$\nu'$ are in the same clock region iff all
clocks have the same integer parts in $\nu$ and~$\nu'$, and if the
partial orders of the clocks, determined by their fractional parts in 
$\nu$ and $\nu'$, are the same.  
For all $\nu \in \Vv$, we write $[\nu]$ for the clock region of
$\nu$. 
A \emph{clock zone} is a convex set of clock valuations, which  
is a union of a set of clock regions. 
Note that a set of clock valuations is a zone iff it is definable by a
clock constraint.
For $W \subseteq \Vv$, we write $\CLOSOP(W)$ for the smallest closed set
in~$\Vv$ which contains~$W$.  
Observe that for every clock zone~$W$, the set $\CLOSOP(W)$ is also a
clock zone.  

Let $L$ be a finite set of \emph{locations}.
A \emph{configuration} is a pair $(\ell, \nu)$, where $\ell \in L$ is
a location and $\nu \in \Vv$ is a clock valuation; 
we write~$Q$ for the set of configurations. 
If $s = (\ell, \nu) \in Q$ and $c \in C$, then we write $s(c)$ for
$\nu(c)$. 
A~\emph{region} is a pair $(\ell, P)$, where $\ell$ is a
location and $P$ is a clock region.  
If $s = (\ell, \nu)$ is a configuration then we write $[s]$ for the
region $(\ell, [\nu])$.
We write~$\Rr$ for the set of regions.
A set $Z \subseteq Q$ is a \emph{zone} if for every $\ell \in L$,
there is a clock zone $W_\ell$ (possibly empty), such that  
$Z = \set{(\ell, \nu) \: : \: 
  \ell \in L \text{ and } \nu \in W_\ell}$.  
For a region $R = (\ell, P) \in \Rr$, we write $\CLOSOP(R)$ for the zone
$\set{(\ell, \nu) \: : \: \nu \in \CLOSOP(P)}$. 

A \emph{timed automaton} $\Tt = (L, C, S, A, E, \delta, \varrho)$ 
consists of 
  a finite set of locations~$L$,
  a finite set of clocks~$C$, 
  a set of \emph{states} $S \subseteq Q$, 
  a finite set of \emph{actions} $A$, 
  an \emph{action enabledness function} $E : A \to 2^S$,  
  a \emph{transition function} $\delta : L \times A \to L$, 
  and a \emph{clock reset function} $\varrho : A \to 2^C$.
We require that $S$, 
and $E(a)$ for all $a \in A$, are zones.
  
Clock zones, from which zones $S$, 
and $E(a)$, for all $a \in A$,
are built, are typically specified by clock constraints. 
Therefore, when we consider a timed automaton as an input of an 
algorithm, its size should be understood as the sum of sizes of
encodings of $L$, $C$, $A$, $\delta$, and $\varrho$, and the sizes of
encodings of clock constraints defining zones $S$, 
and $E(a)$, 
for all $a \in A$.  
Our definition of a timed automaton may appear to differ from the
usual ones~\cite{AD94,BBBR07}, but the differences are superficial. 

For a configuration $s = (\ell, \nu) \in Q$ and $t \in \Rplus$, we
define $s + t$ to be the configuration $s' = (\ell, \nu + t)$ if
$\nu+t \in \Vv$, and we then write $s \xrightharpoonup{}_t s'$. 
We write $s \xrightarrow{}_t s'$ if $s \xrightharpoonup{}_t s'$ and 
for all $t' \in [0, t]$, we have $(\ell, \nu + t') \in S$. 
For an action $a \in A$, we define $\SUCC(s, a)$ to be the
configuration $s' = (\ell', \nu')$, where $\ell' = \delta(\ell, a)$
and $\nu' = \RESET(\nu, \varrho(a))$, and we then write 
$s \xrightharpoonup{a} s'$.
We write $s \xrightarrow{a} s'$ if $s \xrightharpoonup{a} s'$; 
$s, s' \in S$; and $s \in E(a)$.
For technical convenience, and without loss of generality, we will  
assume throughout that for every $s \in S$, there exists $a \in A$,
such that $s \xrightarrow{a} s'$.
For $s, s' \in S$, we say that $s'$ is in the future of $s$, 
or equivalently, that $s$ is in the past of $s'$, 
if there is $t \in \Rplus$, such that $s \xrightarrow{}_t s'$; 
we then write $s \xrightarrow{}_* s'$. 

For $R, R' \in \Rr$, we say that $R'$ is in the future of $R$, 
or that $R$ is in the past of $R'$, if for all $s \in R$, there is 
$s' \in R'$, such that $s'$ is in the future of $s$; 
we then write $R \xrightarrow{}_* R'$.
Similarly, for $R, R' \in \Rr$, we write $R \xrightarrow{a} R'$ if
there is 
$s \in R$, and there is $s' \in R'$, such that $s \xrightarrow{a} s'$.

A \emph{timed action} is a pair $\tau = (t, a) \in \Rplus \times A$.  
For $s \in Q$, we define $\SUCC(s, \tau) = \SUCC(s, (t, a))$ to be the 
configuration $s' = \SUCC(s + t, a)$, i.e., such that 
$s \xrightharpoonup{}_t s'' \xrightharpoonup{a} s'$, and we then write
$s \xrightharpoonup{a}_t s'$.  
We write $s \xrightarrow{a}_t s'$ if 
$s \xrightarrow{}_t s'' \xrightarrow{a} s'$, and we then say that 
$(s, (t, a), s')$ is a \emph{transition} of the timed automaton. 
If $\tau = (t, a)$ then we write $s \xrightharpoonup{\tau} s'$ instead 
of $s \xrightharpoonup{a}_t s'$, and $s \xrightarrow{\tau} s'$ instead
of $s \xrightarrow{a}_t s'$. 

An infinite run of a timed automaton is a sequence 
$r = \seq{s_0, \tau_1, s_1, \tau_2, \dots}$, such that 
for all $i \geq 1$, we have $s_{i-1} \xrightarrow{\tau_i} s_i$. 
A finite run of a timed automaton is a finite sequence 
$\seq{s_0, \tau_1, s_1, \tau_2, \dots, \tau_{n}, s_n} \in S \times ((A
\times \Rplus) \times S)^*$, such that for all $i$, $1 \leq i \leq n$,
we have $s_{i-1} \xrightarrow{\tau_i} s_i$. 
For a finite run 
$r = \seq{s_0, \tau_1, s_1, \tau_2, \dots, \tau_{n}, s_n}$, 
we define $\LENGTH(r) = n$, and we define $\LAST(r) = s_n$ to be the
state in which the run ends. 
For a finite 
run $r = \seq{s_0, \tau_1, s_1, \tau_2, \dots, s_n}$, we define time
of the run as $\TIME(r) = \sum_{i=1}^{n} t_i$. 
We write $\FRUNS$ for the set of finite runs.

\subsection{Strategies}

An average-time game $\Gamma$ is a triple $(\Tt, L_\mMIN, L_\mMAX)$,
where $\Tt = (L, C, S, A, E, \delta, \varrho)$ is a timed automaton and
$(L_\mMIN, L_\mMAX)$ is a partition of $L$. 
We define $Q_\mMIN = \set{ (\ell, \nu) \in Q \::\: \ell \in L_\mMIN}$,
$Q_\mMAX = Q \setminus Q_\mMIN$, 
$S_\mMIN = S \cap Q_\mMIN$, $S_\mMAX = S \setminus S_\mMIN$,
$\Rr_\mMIN = \set{[s] \: : \: s \in Q_\mMIN}$, and 
$\Rr_\mMAX = \Rr \setminus \Rr_\mMIN$. 

A \emph{strategy} for Min is a function 
$\mu : \FRUNS \to A \times \Rplus$, such that if 
$\LAST(r) = s \in S_\mMIN$ and $\mu(r) = \tau$ then 
$s \xrightarrow{\tau} s'$, where $s' = \SUCC(s, \tau)$.
Similarly, a strategy for player Max is a function 
$\chi : \FRUNS \to A \times \Rplus$, such that if 
$\LAST(r) = s \in S_\mMAX$ and $\chi(r) = \tau$ then 
$s \xrightarrow{\tau} s'$, where $s' = \SUCC(s, \tau)$.
We write $\Sigma_\mMIN$ for the set of strategies for player Min, and
we write $\Sigma_\mMAX$ for the set of strategies for player Max.
If players Min and Max use strategies $\mu$ and $\chi$,
resp., then the $(\mu, \chi)$-run from a state $s$ is the
unique run 
$\RUN(s, \mu, \chi) = \seq{s_0, \tau_1, s_1, \tau_2, \ldots}$,  
such that $s_0 = s$, and for every $i \geq 1$, if $s_i \in S_\mMIN$,
or $s_i \in S_\mMAX$, then 
$\mu(\RUN_i(s, \mu, \chi)) = \tau_{i+1}$, 
or $\chi(\RUN_i(s, \mu, \chi)) = \tau_{i+1}$, resp., where 
$\RUN_i(s, \mu, \chi) = 
  \seq{s_0, \tau_1, s_1, \ldots, s_{i-1}, \tau_i, s_i}$.

We say that a strategy $\mu$ for Min is \emph{positional} if
for all finite runs $r, r' \in \FRUNS$, we have that $\LAST(r) =
\LAST(r')$ implies $\mu(r) = \mu(r')$.
A positional strategy for player Min can be then represented as a
function $\mu : S_\mMIN \to A \times \Rplus$, which uniquely
determines the strategy $\mu^\infty \in \Sigma_\mMIN$ as follows:
$\mu^\infty(r) = \mu(\LAST(r))$, for all finite runs  
$r \in \FRUNS$. 
Positional strategies for player Max are defined and represented in
the analogous way. 
We write $\Pi_\mMIN$ and $\Pi_\mMAX$ for the sets of positional
strategies for player Min and for player Max, respectively. 

\subsection{Value of Average-Time Game}
If player Min uses the strategy $\mu \in \SigmaMin$ and player Max
uses the strategy $\chi \in \SigmaMax$ then player Min loses the
value 
\[
\Aa_\mMIN(s, \mu, \chi) = 
\limsup_{n \to \infty} \frac{1}{n} \cdot \TIME(\RUN_n(s, \mu, \chi)),
\]
and player Max wins the value 
\[
\Aa_\mMAX(s, \mu, \chi) = 
\liminf_{n \to \infty} \frac{1}{n}\cdot \TIME(\RUN_n(s, \mu, \chi)).
\] 
In an average-time game player Min is interested in minimising the
value she loses and player Max is interested in maximising the value 
he wins.
For every state $s \in S$ of a timed automaton, we define its 
\emph{upper value} by 
\[
\UVAL^{\Tt}(s) = \inf_{\mu \in \Sigma_\mMIN} \sup_{\chi \in \Sigma_\mMAX}
\Aa_\mMIN(s, \mu, \chi),
\]
and its lower value 
\[
\LVAL^{\Tt}(s) = \sup_{\chi \in \Sigma_\mMAX} \inf_{\mu \in \Sigma_\mMIN}
\Aa_\mMAX(s, \mu, \chi).
\]

The inequality $\LVAL^{\Tt}(s) \leq \UVAL^{\Tt}(s)$ always holds.
An average-time game is \emph{determined} if for every
state $s \in S$, its lower and upper values are equal to each other;  
then we say that the \emph{value} $\VAL^{\Tt}(s)$ exists and 
$\VAL^{\Tt}(s) = \LVAL^{\Tt}(s) = \UVAL^{\Tt}(s)$.
We give an elementary proof for the determinacy of the
average-time games without recourse to general results like
Martin's determinacy theorem~\cite{Mar75,Mar98}. 
\begin{theorem}[Determinacy]
  \label{theorem:atg-are-determined}
  Average-time games are determined.
\end{theorem}

For strategy $\mu \in \Sigma_\mMIN$ of player Min and 
$\chi \in \Sigma_\mMAX$ of player Max, we
define $\VAL^\mu(s) = 
\sup_{\chi \in \Sigma_\mMIN} \Aa_\mMIN(s, \mu, \chi)$, and
$\VAL^\chi(s) = \inf_{\mu \in \Sigma_\mMIN} \Aa_\mMAX(s, \mu, \chi)$.
For an $\varepsilon > 0$, we say that a strategy 
$\mu \in \Sigma_\mMIN$ or $\chi \in \Sigma_\mMAX$ is 
\emph{$\varepsilon$-optimal} if for every $s \in S$ we have that
$\VAL^\mu(s) \leq \VAL^{\Tt}(s) + \varepsilon$ or 
$\VAL^\chi(s) \geq \VAL^{\Tt}(s) - \varepsilon$, respectively.
Note that if a game is determined then for every $\varepsilon > 0$,
both players have $\varepsilon$-optimal strategies.

We say that a strategy $\chi \in \Sigma_\mMAX$ of player Max is a best 
response to a strategy $\mu \in \Sigma_\mMIN$ of player Min if for all 
$s \in S$ we have that 
$\Aa_\mMIN(s, \mu, \chi) = 
\sup_{\chi' \in \Sigma_\mMAX} \Aa_\mMIN(s, \mu, \chi')$.
Similarly we say that a strategy $\mu \in \Sigma_\mMIN$ of player Min
is a best response to a strategy $\chi \in \Sigma_\mMAX$ of player Max 
if for all $s \in S$ we have that 
$\Aa_\mMAX(s, \mu, \chi) =  
\inf_{\mu' \in \Sigma_\mMIN} \Aa_\mMAX(s, \mu', \chi)$.

In the next section we introduce some region-based abstractions of
timed automata, including the closed region graph, and its subgraphs:
the boundary region graph, and the region graph.
While the region graph is semantically equivalent to the corresponding 
timed automaton, the boundary region graph has the property that for 
every starting state, the reachable state space is finite. 
In Section~\ref{section:atg-on-region-graphs} we introduce
average-time games on these graphs and show that if we have the
solution of the average-time game for any of these graphs, then we get
the solution of the average-time game for the corresponding timed
automaton. 
The key Theorem~\ref{theorem:atg-are-determined} follows
immediately from Theorem~\ref{theorem:atg-on-graphs-are-determined}.

\section{Abstractions of Timed Automata}
\label{section:some-useful-abstractions}

The region automaton, originally proposed by Alur and Dill~\cite{AD94}, is
a useful abstraction of a timed automaton as it preserves the validity
of qualitative reachability, safety, and $\omega$-regular properties.  
The \emph{region automaton}~\cite{AD94} $\RegA(\Tt) = (\Rr, \Mm)$ of a
timed automaton $\Tt$ consists of:
\begin{itemize}
\item 
  the set $\Rr$ of regions of~$\Tt$, and 
\item 
  $\Mm \subseteq \Rr \times (\Rr \times A) \times \Rr$, such that 
  for all $a \in A$, and for all $R, R', R'' \in \Rr$, we have that 
  $(R, R'', a, R') \in \Mm$ iff
  $R \xrightarrow{}_{*} R'' \xrightarrow{a} R'$.
\end{itemize}

The region automaton, however, is not sufficient for solving
average-time games as it abstracts away the timing information.  
Corner-point abstraction, introduced by Bouyer et al.~\cite{BBL04}, is a
refinement of region automaton which preserves some timing
information. 
Formally, the corner-point abstraction $\CorA(\Tt)$  of a timed
automaton~$\Tt$ is a finite graph $(V, E)$ such that: 
\begin{itemize}
\item $V \subseteq Q \times \Rr$ such that $(s, R) \in V$ iff 
  $s = (\ell, \nu) \in \CLOSOP(R)$ and $\nu$ is a corner. 
  Since timed automata  we consider are bounded, there are finitely
  many regions, and every region has a finite number of corners.  
  Hence the set of vertices finite.
\item $E \subseteq V \times (\Rplus \times \Rr \times A) \times V$
  such that for $(s, R), (s', R') \in V$ and 
  $(t, R'', a) \in  \Rplus \times \Rr \times A$, we have 
  $((s, R), (t, R'', a), (s', R')) \in E$ 
  iff $R \xrightarrow{}_{*} R'' \xrightarrow{a} R'$
  and $(s+t) \xrightharpoonup{a} s'$.
  Notice that such a $t$ is always a natural number.
\end{itemize}
Bouyer et al.~\cite{BBL04} showed that the corner-point abstraction is 
sufficient for deciding one-player average-price problem if the
initial state is a corner-state, i.e., a state whose clock valuation
is a corner. 

In this section we introduce the \emph{boundary region graph}, which is a
generalisation of the corner-point abstraction. 
We prove that the value of the average-time game on a timed
automaton is equal to the value of the average-time game on the 
corresponding boundary region graph, for all starting states, not just 
for corner states.
In the process, we introduce two other refinements of the region
automaton, 
which we call the \emph{closed region graph} and the 
\emph{region graph}.
We collectively refer to these three graphs as \emph{region graphs}.
The analysis of average-time games on those objects allows us to
establish equivalence of average-time games on the original timed
automaton and the boundary region graph.
We also show (Lemma~\ref{lemma:atg-regionally-constant}) that the
value of an average-time game is constant over a region. 
A side-effect of this result is that the corner-point
abstractions can be used to solve average-time games on timed
automata for arbitrary starting states. 

\subsection{Region Graphs}
A \emph{configuration} in region
graphs is a is a pair $(s, R)$, where 
$s \in Q$ is a configuration of the timed automaton and $R \in \Rr$ is
a region;
We write~$\pQ$ for the set of configurations of the region graphs. 
For a set $X \subseteq \pQ$ and a region $R_0 \in \Rr$, we
define the set $X$ restricted to the region $R_0$ as the set
$\set{(s, R) \in X \::\: R = R_0}$, and we denote this set by $X(R_0)$. 
For a configuration $q = (s, R) \in \pQ$ we write write $[q]$ for its
region~$R$.  

\begin{definition}[Closed Region Graph]
  The \emph{closed region graph}
  $\cTt = (\cS, \cE)$ of a
  timed automaton $\Tt$  is a labelled transition system, where:  
  \begin{itemize}
  \item $\cS$ is the set of \emph{states} defined as 
    \[
    \cS = \set{(s, R) \in  \pQ \::\: s \in \CLOSOP(R)}~\text{
      and}
    \]
  \item $\cE$ is the \emph{labelled transition relation} defined as 
    \begin{multline*}
      \cE = \set{((s, R), (t, R'', a), (s' , R')) \in 
        \cS \times (\Rplus \times \Rr \times A) \times \cS\\
        \::\:  R \xrightarrow{}_{*} R'' \xrightarrow{a} R' \text{ and }
        s' = \SUCC(s, (t, a)) \text{ and }
        s + t \in \CLOSOP(R'')
      }.
    \end{multline*}  
  \end{itemize} 
\end{definition}
\begin{definition}[Boundary Region Graph]
  The \emph{boundary region graph} $\hTt = (\hS, \hE)$ of a
  timed automaton $\Tt$ is a labelled transition system, where:  
  \begin{itemize}
  \item $\hS$ is the set of \emph{states} defined as 
    \[
    \hS = \set{(s, R) \in  \pQ \::\: s \in \CLOSOP(R)}~\text{
      and }
    \]
  \item $\hE$ is the \emph{labelled transition relation} defined as 
    \begin{multline*}
      \hE = \set{((s, R), (t, R'', a), (s' , R')) \in 
        \hS \times (\Rplus \times \Rr \times A) \times \hS\\
        \::\:  R \xrightarrow{}_{*} R'' \xrightarrow{a} R' \text{ and }
        s' = \SUCC(s, (t, a)) \text{ and }
        s + t \in \BOUNDOP(R'')
      }.
    \end{multline*}  
  \end{itemize} 
\end{definition}

Boundary region graphs have the following remarkable property.
\begin{proposition}[\cite{Tri09}]
  \label{proposition:brg-is-finite}
  For every configuration in a boundary region graph the set of 
  reachable configurations is finite.
\end{proposition}

We say that a configuration $q = (s = (\ell, \nu), R)$ is \emph{corner
  configuration} if $\nu$ is a corner. 

\begin{proposition} 
  The reachable sub-graph of the a boundary region graph $\hTt$ from a
  corner configuration is same as the corner-point abstraction
  $\CorA(\Tt)$.  
\end{proposition}

\begin{definition}[Region Graph]
  A \emph{region graph} of a timed automaton $\Tt$ is a labelled
  transition system $\rTt = (\rS, \rE)$, where:
  \begin{itemize}
    
  \item $\rS$ is the set of \emph{states} defined as 
    \[  
    \rS = \set{(s, R) \in \pQ \::\: s \in R}~\text{ and }
    \]
    
  \item $\rE$ is the \emph{labelled transition relation} defined as
    \begin{multline*}
      \rE = \set{((s, R), (t, R'', a), (s' , R')) 
        \in  \rS \times (\Rplus \times \Rr \times A) \times  \rS \\
        \::\: R \xrightarrow{}_{*} R'' \xrightarrow{a} R' \text{ and }
        s' = \SUCC(s, (t, a)) \text{ and }
        s + t \in R''}.
    \end{multline*}
  \end{itemize} 
\end{definition}

For configuration $q = (s, R) \in \pQ$, 
real number $t \in \Rplus$, region $R'' \in \Rr$, and action 
$a \in A$, we write $\SUCC(q, (t, R'', a))$ for the configuration 
$\big(\SUCC(s, (t, a)), R'\big)$ where $R'' \xrightarrow{a} R'$. 

\subsection{Region Game Graphs}
For $\Gamma = (\Tt, L_\mMIN, L_\mMAX)$ we
define the sets $\pQ_\mMIN = \set{(s, R) \in \pQ \::\: R \in \Rr_\mMIN}$
and $\pQ_\mMAX = \pQ \setminus \pQ_\mMIN$.
Similarly we define sets $\cS_\mMIN$, $\cS_\mMAX$, $\hS_\mMIN$,
$\hS_\mMAX$, $\rS_\mMIN$, and $\rS_\mMAX$.
The timed game automaton $\Gamma$  naturally gives rise to
the closed region game graph 
$\cGg = (\cTt, \cS_\mMIN, \cS_\mMAX)$, 
the boundary region game graph 
$\hGg = (\hTt, \hS_\mMIN, \hS_\mMAX)$, and  
the region game graph $\rGg = (\rTt, \rS_\mMIN, \rS_\mMAX)$. 
When it is clear from context, we use the terms region graphs and
region game graphs interchangeably.
Also, sometimes, we write $\Tt$, $\cTt$, $\hTt$, and
$\rTt$ for $\Gamma$, $\cGg$, $\hGg$, and $\rGg$, respectively.

\subsection{Runs of Region Graphs} 
An infinite run of the closed region graph $\cTt = (\cS, \cE)$ is an
infinite sequence 
\[
\seq{q_0, \tau_1, q_1, \tau_1, \ldots} \in \cS \times \big((\Rplus \times
\Rr \times A) \times \cS \big)^\omega,
\] 
such that for every positive integer $i$ we have 
$(q_{i-1}, \tau_i, q_i) \in \cE$. 
A finite run of the closed region graph $\cTt$ is a finite sequence  
\[
\seq{q_0, \tau_1, q_1, \tau_1, \ldots, q_n} \in 
\cS \times \big((\Rplus \times \Rr \times A) \times \cS\big)^*,
\] 
such that for every positive integer $i \leq n$ we have 
$(q_{i-1}, \tau_i, q_i) \in \cE$. 
Runs of the boundary region graph and the region graph are defined
analogously.  

For a graph $\Gg \in \set{\cTt, \hTt, \rTt}$ we write $\RUNS^{\Gg}$ 
for the set of its runs and $\RUNS^{\Gg}(q)$ for the set
of its runs from a state $q \in \cQ$.  
We write $\FRUNS^{\Gg}$ for the set of finite runs and
$\FRUNS^{\Gg}(q)$  for the set of finite runs starting from $q \in \cS$.

\subsection{Pre-Runs and Run Types}
Pre-runs~\cite{JT08} generalise runs of $\cTt, \rTt$, and $\hTt$, and
allow us to compare the runs in $\cTt, \rTt,$ and $\hTt$ in a uniform
manner.  
On the other hand, the concept of the type~\cite{JT08} of a run allows
us to compare pre-runs passing through the same sequence of regions. 

A \emph{pre-run} is a sequence 
$\seq{(s_0, R_0), (t_1, R'_1, a_1), (s_1, R_1), \dots}
\in \pQ \times ((\Rplus \times \Rr \times A) \times \pQ)^\omega$,
such that $s_{i+1} = \SUCC(s_i, (t_{i+1}, a_{i+1}))$ and 
$R_i \xrightarrow{}_{*} R'_{i+1} \xrightarrow{a_{i+1}} R_{i+1}$ for every 
$i \in \Nat$.
We write $\PRERUNS$ for the set of pre-runs and $\PRERUNS(s, R)$ for
the set of pre-runs starting from $(s, R) \in \pQ$.
The relation between various sets of runs is as follows:~for all $q
\in \cQ$ we have
\begin{eqnarray*}
  \RUNS^{\hTt}(q) \subseteq \RUNS^{\cTt}(q) \subseteq \PRERUNS(q)
  ~\text{ and } \\ 
  \RUNS^{\rTt}(q) \subseteq \RUNS^{\cTt}(q) \subseteq \PRERUNS(q).
\end{eqnarray*}

A finite pre-run is a finite sequence
$\seq{(s_0, R_0), (t_1, R'_1, a_1), \dots, (s_n, R_n)}
\in (Q \times \Rr) \times ((\Rplus \times \Rr \times A) \times 
(Q \times \Rr))^*$
such that for every nonnegative integer $i < n$ we have that
$s_{i+1} = \SUCC(s_i, (t_{i+1}, a_{i+1}))$ and 
$R_i \xrightarrow{}_{*} R'_i \xrightarrow{a_{i+1}} R_i$.
We write $\FPRERUNS$ for the set of finite pre-runs and 
$\FPRERUNS(s, R)$ for the set of finite pre-runs starting from 
$(s, R) \in \pQ$. 
For finite run 
$r = \seq{q_0, (t_1, R_1, a_1), q_1, \ldots, q_n} \in
\FPRERUNS$ we define its total time as $\TIME(r) = \sum_{i=1}^n t_i$,
and we denote the last state of the run by $\LAST(r) = q_n$.  

A \emph{run type} is a sequence 
$\seq{R_0, (R'_1, a_1), R_1, (R'_2, a_2), \dots} \in 
\Rr \times ((\Rr \times A) \times \Rr)^\omega$
such that for every $ i \in \Nat$ we have that
$R_i \xrightarrow{}_{*} R'_{i+1} \xrightarrow{a} R_{i+1}$.
We say that a pre-run 
$r = \seq{(s_0, R_0), (t_1, R'_1, a_1), (s_1, R_1), 
  (t_1, R'_2, a_2), \dots}$ is of the type 
$\seq{R_0, (R'_1, a_1), R_1, (R'_2, a_2), \dots}$.
We say that a run $r = \seq{s_0, (t_1, a_1), s_1, (t_2, a_2), \dots}$
of a timed automaton~$\Tt$ is of the type 
$\seq{R_0, (R'_1, a_1), R_1, (R'_2, a_2), \dots}$, 
where $R_i = [s_i]$ and 
$R'_{i+1} = [s_i + t_{i+1}]$ for all $i \in \Nat$.
We also define the type of a finite runs analogously. 

For a (finite or infinite) run or pre-run $r$, we write $\TYPE{r}$ for
its type. 
We write $\TYPES$ for the set of run types, and we write $\TYPES(R)$  
for the set of run types starting from region $R \in \Rr$. 
Similarly we write $\FTYPES$ for the set of finite run types, and we write 
$\FTYPES(R)$  for the set of finite run types starting from region 
$R \in \Rr$. 

\section{Strategies in Region Graphs}
In this section we define strategies of players in region graphs 
$\cTt$, $\rTt$, and $\hTt$, and study some of their properties.
Strategies in $\rTt$ are called \emph{admissible strategies}, while
strategies in $\hTt$ are called \emph{boundary strategies}.
We also introduce so-called \emph{type-preserving boundary strategies}
which are a key tool in proving the correctness of game reduction from
timed automata to boundary region graph.
In Section~\ref{section:atg-on-region-graphs} we show that there are
optimal type-preserving boundary strategies in $\cTt$ and $\hTt$.

\subsection{Pre-strategies and Strategies in $\cTt, \hTt, \rTt$}

Pre-strategies generalise the concept of strategies in region graphs,
and allows us to discuss the strategies in $\cTt$, $\hTt$, and $\rTt$
in a uniform manner.
We first define pre-strategies for players in $\Tt$, and then using
that we define strategies for players in closed region graph,
boundary region graph, and region graph. 

\begin{definition}[Pre-strategies]
  A \emph{pre-strategy} of player Min $\mu$ is a (partial) function 
  $\mu: \FPRERUNS \to \Rplus \times \Rr \times A$, such that 
  for a run $r \in \FPRERUNS$, 
  if $\LAST(r) = (s, R) \in \pQ_\mMIN$  
  then $\mu(r) = (t, R'', a)$ is defined, and it is such that 
  $R \xrightarrow{}_{*} R'' \xrightarrow{a} R'$ for some 
  $R' \in \Rr$.
  Pre-strategies of player Max are defined analogously.  
  We write $\pSigmaMin$ and $\pSigmaMax$ for the set of pre-strategies
  of player Min and player Max, respectively.
\end{definition}

We say that a strategy of player Min $\mu \in \pSigmaMin$ is
positional if for all runs $r_1, r_2 \in \FPRERUNS$ we have that
$\LAST(r_1) = \LAST(r_2)$ implies $\mu(r_1) = \mu(r_2)$.
Similarly we define positional strategy of player Max.

We define the run starting from configuration $q \in \pQ$ where player
Min and player Max use the strategies $\mu \in \pSigmaMin$ and 
$\chi \in \pSigmaMax$, respectively, in a straightforward manner and we
write $\RUN(q, \mu, \chi)$ for this run.  
For every positive integer $n$ we write $\RUN_n(q, \mu, \chi)$ for the
prefix of the run $\RUN(q, \mu, \chi)$ of length $n$. 

Now we are in a position to introduce strategies in closed region
graph, region graph, and boundary region graph.

\begin{definition}[Strategies in Closed Region Graph]
  A pre-strategy of player Min ${\mu \in \pSigmaMin}$ is a strategy in a
  closed region graph $\cTt = (\cS, \cE)$ if for every run 
  $r \in \FPRERUNS$ such that $\mu(r) = (t, R', a)$, we have that 
  $(s+t) \in \CLOSOP(R')$ where $(s, R) = \LAST(r)$.
  Strategies of player Max in a closed region graph
  are defined analogously.   
  We write $\cSigmaMin$ and $\cSigmaMax$ for the set of strategies of
  player Min and player Max, respectively.
\end{definition}

\begin{definition}[Strategies in Region Graphs]
  A pre-strategy of player Min $\mu \in \pSigmaMin$ is a strategy in a
  region graph $\rTt = (\rS, \rE)$ if for every run 
  $r \in \FRUNS^{\rTt}$ such that $\mu(r) = (t, R'', a)$, we have that 
  $(s+t) \in R''$ where $(s, R) = \LAST(r)$.
  Strategies of player Max in a region 
  graph are defined analogously.
  We call such strategies \emph{admissible strategies}.
  We write $\rSigmaMin$ and $\rSigmaMax$ for the set of admissible
  strategies of player Min and player Max, respectively.
\end{definition}

\begin{definition}[Strategies in Boundary Region Graph]
  A pre-strategy of player Min $\mu \in \pSigmaMin$ is a strategy in a
  boundary region graph $\hTt = (\hS, \hE)$ if for every run 
  $r \in \FPRERUNS$ such that $\mu(r) = (t, R', a)$, we have that 
  \begin{equation}
    t = \inf \set{ t \::\: s + t \in \CLOSOP(R')},\label{e:brg-strat-inf} 
  \end{equation}
  where $(s, R) = \LAST(r)$.\\
  A pre-strategy of player Max $\chi\in \pSigmaMax$ is a 
  strategy in a
  boundary region graph $\hTt$ if for every run $r \in \FPRERUNS$ such
  that $\mu(r) = (t, R', a)$, we have that  
  \begin{equation}
    t = \sup \set{ t \::\: s + t \in \CLOSOP(R')},\label{e:brg-strat-sup} 
  \end{equation}
  where $(s, R) = \LAST(r)$.
  We call such strategies \emph{boundary strategies}.
  We write $\hSigmaMin$ and $\hSigmaMax$ for the set of boundary
  strategies of player Min and player Max, respectively.
\end{definition}

For notational convenience and w.l.o.g., in the definition of boundary
strategies, we do not consider those timed moves of player Min (Max)
which suggest waiting till the farther (nearer) boundary of a thick
region. 

\begin{remark}
  For every state $s \in S$ of timed automata $\Tt$ and every strategy
  $\mu \in \pSigmaMin$ and $\chi \in \pSigmaMax$ of respective
  players, we have that :
  \begin{itemize}
  \item $\RUN((s, [s]), \mu, \chi) \in \RUNS^{\cTt}(s, [s])$ if $\mu
    \in \cSigmaMin$ and $\chi \in \cSigmaMax$;
  \item $\RUN((s, [s]), \mu, \chi) \in \RUNS^{\hTt}(s, [s])$ if $\mu
    \in \hSigmaMin$ and $\chi \in \hSigmaMax$;
  \item$\RUN((s, [s]), \mu, \chi) \in \RUNS^{\rTt}(s, [s])$ if $\mu
    \in \rSigmaMin$ and $\chi \in \rSigmaMax$.
  \end{itemize}
\end{remark}

\subsubsection{Boundary Strategies and Boundary Timed Actions.} 
Define the finite set of \emph{boundary timed
  actions}
$\aA = \NATS{k} \times C \times A$.
For $s \in Q$ and $\alpha = (b, c, a) \in \aA$, we
define $t(s, \alpha) = b - s(c)$ if $s(c) \leq b$, and 
$t(s, \alpha) = 0$ if $s(c) > b$; 
and we define $\SUCC(s, \alpha)$ to be the state 
$s' = \SUCC(s, \tau(\alpha))$, 
where $\tau(\alpha) = (t(s, \alpha), a)$;
we then write $s \xrightharpoonup{\alpha} s'$.
We also write $s \xrightarrow{\alpha} s'$ if 
$s \xrightarrow{\tau(\alpha)} s'$. 
For configuration $q = (s, R) \in \pQ$, 
boundary timed action $\alpha = (b, c, a) \in \aA$, and region 
$R'' \in \Rr$ we write $\SUCC(q, (\alpha, R''))$ for
the configuration $\SUCC(q, (t(s, \alpha), R'', a))$.

Timed actions suggested by a boundary strategies are precisely
boundary timed actions.
The following proposition formalises this notion.
\begin{proposition}
  \label{prop:boundary-strategies-timed-moves}
  For every boundary strategy  $\sigma \in \hSigmaMin (\hSigmaMax)$ 
  of player Min (Max) and for every run 
  $r \in \FPRERUNS$, if $\sigma(r) = (t, R', a)$ then there exists
  a boundary timed action $\alpha = (b, c, a) \in \aA$ such that 
  $t(s, \alpha) = t$, where $(s, R) = \LAST(r)$. 
\end{proposition}
\begin{proof}
  Let run $r \in \FPRERUNS$ be such that $\LAST(r) = (s, R)$. 
  Let $\sigma \in \hSigmaMin$ be a boundary strategy of player Min
  such that $\sigma(r) = (t, R', a)$. 
  From the definition of the boundary strategies, we have that 
  $t = \inf \set{ t \::\: s + t   \in \CLOSOP(R')}$.
  To prove the proposition, all we need to show is that there exists
  an integer $b \in \Int$ and a clock $c \in C$, such
  that $b - s(c) = t$. 
  
  If $R' \in \Rr_\THIN$ then there exists a clock $c' \in C$ such that
  for all states ${s' \in \CLOSOP(R')}$ we have that $\FRAC{s'(c')} = 0$. 
  In this case the clock $c = c'$ and the integer $b = (s+t)(c)$. 

  If $R \in \Rr_\THICK$ and let $R' \xleftarrow{}_{+1} R$ be the thin
  region immediately before $R$.
  Let clock $c' \in C$ be such that for all states 
  $s' \in \CLOSOP(R')$ we have that $\FRAC{s'(c')} = 0$.
  Again, in this case the desired clock $c = c'$ and the integer 
  $b = (s+t)(c)$. 
  
  The case, where $\sigma$ is a strategy of Max is similar, and hence
  omitted. 
  \qed
\end{proof}

Sometimes, in our proofs we need to use boundary timed action
suggested by a boundary strategy. 
For this purpose we define the notation $\BS{\sigma}(r)$ that 
gives the boundary timed action and region pair that corresponds to
$\sigma(r)$. 
The definition of this function is formalised in the following
definition. 
\begin{definition}
  For a boundary strategy $\sigma \in \hSigmaMax (\hSigmaMax)$ of player
  Min (Max), we define the function 
  $\BS{\sigma} : \FPRERUNS \to (\aA \times \Rr)$ as follows:~if for a
  run $r \in \FPRERUNS$ we have $\sigma(r) = (t, R', a)$, then 
  $\BS{\sigma}(r) = ((b, c, a), R')$ such that $b - s(c) = t$, where 
  $(s, R) = \LAST(r)$.
\end{definition}

\subsection{Type-Preserving Boundary Strategies}
\label{subsection:type-preserving}
We now introduce an important class of boundary
strategies called \emph{type-preserving boundary strategies}.
Broadly speaking, these strategies suggest to players a unique
boundary timed action and region pair for all the finite runs of the 
same type. 
\begin{definition}[Type-Preserving Boundary Strategies]
  A boundary strategy $\sigma \in \hSigmaMin$ of player
  Min is \emph{type-preserving}
  if $\TYPE{r_1} = \TYPE{r_2}$ implies  
  $\BS{\sigma}(r_1) = \BS{\sigma}(r_2)$ for all $r_1, r_2
  \in \FPRERUNS$. 
  Type-preserving boundary strategies of player Max are defined
  analogously. 
  We write $\simpleMin$ and $\simpleMax$ for the sets of type-preserving
  boundary strategies of players Min and Max, respectively.
\end{definition}
The rationale behind the name \emph{type-preserving} is that  
if $\mu \in \simpleMin$ and
$\chi \in \simpleMax$, then for every $R \in \Rr$ and for  
$q, q' \in \pQ(R)$, the run types of the resulting runs from $q$ and
$q'$ are the same, i.e., 
$\TYPE{\RUN(q, \mu, \chi)} = \TYPE{\RUN(q', \mu, \chi)}$.

\subsubsection{Simple Functions.}
Let $X \subseteq \pQ$.
A function $F : \cQ \to \Real$ is \emph{simple}~\cite{AM99,JT07}
if either:~there is $e \in \Int$, such that for every $q = (s, R) \in X$, we
have  $F(q) = e$; or 
there are $e \in \Int$ and $c \in C$, such that for every 
$q = (s, R) \in X$, we have $F(q) = e - s(c)$.   
We say that a function $F: X \to \Real$ is \emph{regionally simple}
or  \emph{regionally constant}, respectively, 
if for every region $R \in \Rr$, the function $F$, over domain $X(R)$,
is simple or constant, respectively.

For regions $R, R', R'' \in \Rr$ and boundary timed action 
$\alpha = (b, c, a) \in \aA$, we write 
$R \xrightarrow{R''}_\alpha R'$ if one of the following holds:
\begin{itemize}
\item 
  $R \xrightarrow{}_{b, c} R'' \xrightarrow{a} R'$, or
\item
  there is region $R''' \in \Rr_\THIN$ 
  such that $R \xrightarrow{}_{b, c} R''' \xrightarrow{}_{+1}
  R'' \xrightarrow{a} R'$, or  
\item
  there is a region $R''' \in \Rr_\THIN$ 
  such that $R \xrightarrow{}_{b, c} R''' \xleftarrow{}_{+1} 
  R''\xrightarrow{a} R'$.
\end{itemize}

\subsubsection{Properties of Type-preserving Boundary Strategies.}
The next two proposition state that if
both players play with type-preserving boundary strategies then for
every $n \in \Nat$ the total time spent in $n$ transitions is
regionally simple
(Proposition~\ref{proposition:time-is-regionally-simple}), and 
the average time of the infinite run is regionally constant
(Proposition~\ref{prop:rcbs-are-regionally-constant}).  
\begin{proposition}[Type-preserving strategy pairs yield regionally
  simple time for finite runs]
  \label{proposition:time-is-regionally-simple}
  If $\mu \in \simpleMin$, $\chi \in \simpleMax$, and $n \in \Nat$,
  then $\TIME(\RUN_n(\cdot, \mu, \chi)): \cQ \to \Rplus$ is
  regionally simple.  
\end{proposition}

\begin{proposition}[Type-preserving strategy pairs yield regionally
  constant average time]
  \label{prop:rcbs-are-regionally-constant}
  If $\mu \in \simpleMin$ and $\chi \in \simpleMax$ then
  $\AVERAGETIME_\mMIN(\cdot, \mu, \chi): \cQ \to \Rplus$ and 
  $\AVERAGETIME_\mMAX(\cdot, \mu, \chi): \cQ \to \Rplus$ are
  regionally constant.  
\end{proposition}

\subsubsection{Type-preserving Boundary Strategy that Agrees with a  Boundary Strategy.} 
Given an arbitrary boundary strategy $\sigma$ and a configuration 
$q \in \cQ$, sometimes we are interested in a type-preserving boundary
strategy that agrees with $\sigma$ for all the runs starting from $q$.
We denote such a strategy by $\tpq{\sigma}{q}$.
The following definition formalises such strategy.
\begin{definition}
  \label{definition:tpbs-agree-bs}
  For a boundary strategy $\mu \in \hSigmaMin$ of player Min 
  and $q \in \cQ$ we define 
  $\tpq{\mu}{q} \in \simpleMin$ to be a type-preserving boundary
  strategy which satisfy the following conditions:
  \begin{enumerate}
  \item $\BS{\tpq{\mu}{q}}(r) = \BS{\mu}(r)$ for every 
    $r \in \FPRERUNS(q)$, and 
  \item $\TYPE{r} = \TYPE{r'}$ implies $\BS{\tpq{\mu}{q}}(r) =
    \BS{\tpq{\mu}{q}}(r')$ for all runs $r, r' \in \FPRERUNS$.
  \end{enumerate}
  For $\chi \in \hSigmaMax$ and $q \in \cQ$ we define 
  $\tpq{\chi}{q} \in \simpleMax$ analogously.
\end{definition}

Given an arbitrary strategy $\mu \in \cSigmaMin$ of player Min, a
type-preserving boundary strategy $\chi \in \simpleMax$ of player Max,
and a configuration $q \in \cQ$ sometimes we require to specify a
type-preserving strategy ${\tpqs{\mu}{q}{\chi} \in \simpleMin}$ which
has the property that types of runs $\RUN(q, \mu, \chi)$ and 
$\RUN(q, \tpqs{\mu}{q}{\chi}, \chi)$ are the same.
We then argue that from configuration $q \in \cQ$ if player Max plays
according to $\chi \in \simpleMax$ then player Min can achieve better
average-time if she plays according to $\tpqs{\mu}{q}{\chi}$ (see
Proposition~\ref{proposition:nstep-type-preserving-is-better} and
Corollary~\ref{corollary:type-preserving-is-better}).  
The motivation for the definition of
$\tpqs{\chi}{q}{\mu}$ is similar. 
\begin{definition}
  \label{definition:tp-better-than-normal}
  For an arbitrary strategy $\mu \in \cSigmaMin$ of player Min, a
  type-preserving boundary strategy $\chi \in \simpleMax$ of 
  player Max, and a configuration $q = (s, R) \in \cQ$, we define 
  $\tpqs{\mu}{q}{\chi} \in \simpleMin$ to be a type-preserving boundary
  strategy which satisfy the following conditions:
  \begin{enumerate}
  \item $\TYPE{\RUN(q, \tpqs{\mu}{q}{\chi}, \chi)} 
    = \TYPE{\RUN(q, \mu, \chi)}$, and
  \item $\TYPE{r} = \TYPE{r'}$ implies $\BS{\tpqs{\mu}{q}{\chi}}(r) =
    \BS{\tpqs{\mu}{q}{\chi}}(r')$ for all runs $r, r' \in \FPRERUNS$.
  \end{enumerate}
  For $\chi \in \cSigmaMax$, $\mu \in \simpleMin$, and $q \in \cQ$ 
  the strategy $\tpqs{\chi}{q}{\mu} \in \simpleMax$
  is defined analogously.
\end{definition}

The following proposition and its corollary shows that starting from a
configuration $q$ player Min (Max) prefers $\tpqs{\mu}{q}{\chi}$ 
($\tpqs{\chi}{q}{\mu}$) to $\mu$ ($\chi$) against a type-preserving
strategy $\chi \in \simpleMax$ ($\mu \in \simpleMin$) of its opponent. 
\begin{proposition}
  \label{proposition:nstep-type-preserving-is-better}
  For every $\chi \in \simpleMax$, $\mu \in \cSigmaMin$ and
  $q \in \cQ$ we have that 
  \[
  \TIME(\RUN_n(q, \mu, \chi)) \geq  
  \TIME(\RUN_n(q, \tpqs{\mu}{q}{\chi}, \chi)), 
  \]
  for every $n \in \Nat$.
  Similarly, for every $\mu \in \simpleMin$, $\chi \in \cSigmaMax$ and
  $q \in \cQ$ we have that 
  \[
  \TIME(\RUN_n(q, \mu, \chi)) \leq 
  \TIME(\RUN_n(q, \mu, \tpqs{\chi}{q}{\mu})),
  \]
  for every $n \in \Nat$.
\end{proposition}

An easy corollary of this proposition is as follows:
\begin{corollary}
  \label{corollary:type-preserving-is-better}
  For every $\chi \in \simpleMax$, $\mu \in \cSigmaMin$ and for all
  configurations $q \in \cQ$ we have that
  \[
  \AVERAGETIME_\mMIN(q, \mu, \chi)) \geq
  \AVERAGETIME_\mMIN(q, \tpqs{\mu}{q}{\chi}, \chi)).
  \]
  Similarly for every $\mu \in \simpleMin$, $\chi \in \cSigmaMax$ and
  for all configurations $q \in \cQ$ we have that
  \[
  \AVERAGETIME_\mMAX(q, \mu, \chi)) \leq
  \AVERAGETIME_\mMAX(q, \mu, \tpqs{\chi}{q}{\mu})).
  \]
\end{corollary}

\subsubsection{Admissible Strategies $\varepsilon$-Close to  a Type-Preserving Boundary Strategy.}
Given a type-preserving boundary strategy $\sigma$ and a positive real
$\varepsilon > 0$, sometimes we are interested in admissible
strategies that behave like $\sigma$ within $\varepsilon$ precision.
The following definition formalises such strategy.
\begin{definition}
  \label{def:epsilon-close-strat}
  For $\mu \in \simpleMin$ and a real number $\varepsilon > 0$, we
  define the set of admissible strategy $\EcMin{\mu}{\varepsilon} 
  \subseteq \rSigmaMin$ as
  follows.
  For every $\mu_\varepsilon \in \EcMin{\mu}{\varepsilon}$ we have
  that for all runs $r \in \FPRERUNS$ 
  if $\BS{\mu}(r) = ((b, c, a), R')$ then 
  $\mu_\varepsilon(r) = (t, R', a)$ is such that
  \[
  s + t \in R' \text{ and } t \leq b - s(c) + \varepsilon,
  \]
  where $(s, R) = \LAST(r)$. 
  Notice that (see Equation~\ref{e:brg-strat-inf}) such a value of $t$
  always exists.  \\
  Similarly for $\chi \in \simpleMax$ and a real number $\varepsilon > 0$ we
  define the set $\EcMax{\chi}{\varepsilon} \subseteq \rSigmaMax$ as
  follows.
  For every $\chi_\varepsilon \in \EcMax{\chi}{\varepsilon}$ we have
  that for all runs $r \in \FPRERUNS$ 
  if $\BS{\chi}(r) = ((b, c, a), R')$ then 
  $\chi_\varepsilon(r) = (t, R', a)$ is such that
  \[
  s + t \in R' \text{ and } t \geq b - s(c) - \varepsilon,
  \]
  where $(s, R) = \LAST(r)$. 
\end{definition}

Given an arbitrary strategy $\mu \in \cSigmaMin$ of player Min, a
positive real $\varepsilon > 0$, a type-preserving boundary strategy
$\chi \in \simpleMax$ of player Max, an $\varepsilon$-close strategy
$\chi_\varepsilon \in \EcMax{\chi}{\varepsilon}$,
and a configuration $q \in \cQ$
sometimes we require to specify a 
type-preserving strategy ${\tpqs{\mu}{q}{\chi_\varepsilon} \in
\simpleMin}$ which has the property that types of runs 
$\RUN(q, \mu, \chi_\varepsilon)$ and 
$\RUN(q, \tpqs{\mu}{q}{\chi_\varepsilon}, \chi_\varepsilon)$ are the 
same. 
\begin{definition}
  For an arbitrary strategy $\mu \in \cSigmaMin$ of player Min,  
  a positive real $\varepsilon > 0$,  
  a type-preserving boundary strategy ${\chi \in \simpleMax}$ of player
  Max, 
  an $\varepsilon$-close strategy 
  ${\chi_\varepsilon \in \EcMax{\chi}{\varepsilon}}$, 
  and a configuration $q = (s, R) \in \cQ$, we define 
  $\tpqs{\mu}{q}{\chi_\varepsilon} \in \simpleMin$ to be a
  type-preserving boundary strategy which satisfy the following 
  conditions: 
  \begin{enumerate}
  \item $\TYPE{\RUN(q, \tpqs{\mu}{q}{\chi_\varepsilon}, \chi_\varepsilon)} 
    = \TYPE{\RUN(q, \mu, \chi_\varepsilon)}$, and
  \item $\TYPE{r} = \TYPE{r'}$ implies $\BS{\tpqs{\mu}{q}{\chi_\varepsilon}}(r) =
    \BS{\tpqs{\mu}{q}{\chi_\varepsilon}}(r')$ for all runs 
    $r, r' \in \FPRERUNS$.
  \end{enumerate}
  Combining it with Definition~\ref{def:epsilon-close-strat} we get
  that $\TYPE{\RUN(q, \tpqs{\mu}{q}{\chi_\varepsilon}, \chi)} 
  = \TYPE{\RUN(q, \mu, \chi_\varepsilon)}$.\\
  For $\chi \in \cSigmaMax$, $\chi_\varepsilon \in
  \EcMax{\chi}{\varepsilon}$,
  $\mu \in \simpleMin$, and $q \in \cQ$ 
  the strategy $\tpqs{\chi}{q}{\mu_\varepsilon} \in \simpleMax$
  is defined analogously.
\end{definition}
We need the following property of $\tpqs{\mu}{q}{\chi_\varepsilon}$ and
$\tpqs{\chi}{q}{\mu_\varepsilon}$ strategies.
\begin{proposition}
  \label{proposition:nstep-boundary-better-against-epsilon}
  For every arbitrary strategy $\mu \in \cSigmaMin$, 
  positive real $\varepsilon >0$,
  type-preserving boundary strategy $\chi \in \simpleMax$ of
  player Max, 
  $\varepsilon$-close strategy 
  $\chi_\varepsilon \in \EcMax{\chi}{\varepsilon}$ of player Max,
  and $q \in \cQ$ we have  
  \[
  \TIME(\RUN_n(q, \mu, \chi_\varepsilon)) \geq  
  \TIME(\RUN_n(q, \tpqs{\mu}{q}{\chi_\varepsilon}, \chi)) 
  - n\cdot\varepsilon,
  \] 
  for every $n \in \Nat$.
  Similarly for every arbitrary strategy $\chi \in \cSigmaMax$, 
  positive real $\varepsilon >0$,
  type-preserving boundary strategy $\mu \in \simpleMin$ of
  player Max, 
  $\varepsilon$-close strategy 
  $\mu_\varepsilon \in \EcMin{\mu}{\varepsilon}$ of player Min,
  and $q \in \cQ$ we have  
  \[
  \TIME(\RUN_n(q, \mu_\varepsilon, \chi)) \leq  
  \TIME(\RUN_n(q, \mu_\varepsilon, \tpqs{\chi}{q}{\mu_\varepsilon})) 
  + n\cdot\varepsilon,
  \] 
  for every $n \in \Nat$.  
\end{proposition}

The following result is an easy corollary of 
Proposition~\ref{proposition:nstep-boundary-better-against-epsilon}. 
\begin{corollary}
  \label{corollary:boundary-better-against-epsilon}
  For every $\chi \in \simpleMax$, $\mu \in \cSigmaMin$, 
  $\varepsilon > 0$, $\chi_\varepsilon \in \EcMax{\chi}{\varepsilon}$,
  and $q \in \cQ$ we have that
  \[
  \AVERAGETIME_\mMAX(q, \mu, \chi_\varepsilon)) \geq  
  \AVERAGETIME_\mMAX(q, \tpqs{\mu}{q}{\chi_\varepsilon}, \chi)) - \varepsilon.
  \] 
  Similarly for every $\mu \in \simpleMin$, $\chi \in \cSigmaMax$, 
  $\varepsilon > 0$, $\mu_\varepsilon \in \EcMin{\mu}{\varepsilon}$,
  and $q \in \cQ$ we have that  
  \[
  \AVERAGETIME_\mMIN(q, \mu_\varepsilon, \chi) \leq 
  \AVERAGETIME_\mMIN(q, \mu, \tpqs{\chi}{q}{\mu_\varepsilon}) + \varepsilon.
  \] 
\end{corollary}

To summarise the relations between various strategies, note that the
following inclusions hold:
\begin{eqnarray*}
  \simpleMin \subseteq \hSigmaMin \subseteq \cSigmaMin \subseteq \pSigmaMin
  & \text{ and }&
  \rSigmaMin \subseteq \cSigmaMin \subseteq \pSigmaMin,
  \hspace{1em} \text{and} \hspace{1em}\\
  \simpleMax \subseteq \hSigmaMax \subseteq \cSigmaMax \subseteq \pSigmaMax
  & \text{ and } &
  \rSigmaMax \subseteq \cSigmaMax \subseteq \pSigmaMax.
\end{eqnarray*}

\section{Average-Time Games on Region Graphs}
\label{section:atg-on-region-graphs}

We define
$\AVERAGETIME_\mMIN: \pQ \times \pSigmaMin \times \pSigmaMax \to
\Rplus$ and 
$\AVERAGETIME_\mMAX: \pQ \times \pSigmaMin \times \pSigmaMax \to
\Rplus$ in the following manner:
\begin{eqnarray*}
  \AVERAGETIME_\mMIN(q, \mu, \chi) &=& \limsup_{n \to \infty} \frac{1}{n}\cdot
  \TIME(\RUN_n(q, \mu, \chi)),~\text{ and }\\
  \AVERAGETIME_\mMAX(q, \mu, \chi) &=& \liminf_{n \to \infty} \frac{1}{n}\cdot
  \TIME(\RUN_n(q, \mu, \chi)),
\end{eqnarray*}
where $\mu \in \pSigmaMin$, $\chi \in \pSigmaMax$ and $q \in \pQ$.
For average-time games on a graph $\Gg \in \set{\cTt, \hTt, \rTt}$ we
define the lower-value $\LVAL^{\Gg}(q)$, the upper-value
$\UVAL^{\Gg}(q)$ and the value $\VAL^{\Gg}(q)$ of a configuration 
$q \in \cQ$ in a straightforward manner. 

From construction it clear that the difference between an
average-time game on a timed automaton and the average-time game on
corresponding region graph is purely syntactical.
Hence if the average-time game on region graph $\rTt$ is determined then
average-time game on timed automaton $\Tt$ is determined as well. 

\begin{proposition}
  \label{proposition:region-graph-is-timed-automaton}
  An average-time game on timed automaton $\Tt$ is determined, if
  the corresponding average-time game on region graph $\rTt$ is
  determined. 
  Moreover for all $s \in S$ we have that $\VAL(s) = 
  \VAL^{\rTt}(s, [s])$.
\end{proposition}

The following is the main result of this section.
\begin{theorem}
  \label{theorem:atg-on-graphs-are-determined} 
  Let $\Tt$ be a timed automaton.
  Average-time games on the timed automaton~$\Tt$, the closed region
  graph~$\cTt$, the region graph~$\rTt$, and the boundary region graph
  $\hTt$ are determined. 
  Moreover for every~$s \in S$ in a timed automaton~$\Tt$, we have:
  \[
  \VAL^{\Tt}(s) = \VAL^{\rTt}(s, [s]) = \VAL^{\cTt}(s, [s]) =
  \VAL^{\hTt}(s, [s]). 
  \]
\end{theorem}

This theorem follows from
Theorem~\ref{lemma:boundary-region-graph-determined},
Theorem~\ref{theorem:optimal-xi-in-bar}, 
Theorem~\ref{theorem:region-graph-is-determined}, and
Proposition~\ref{proposition:region-graph-is-timed-automaton}.

Moreover Theorem~\ref{theorem:atg-on-graphs-are-determined} and
Proposition~\ref{prop:rcbs-are-regionally-constant} let us conclude
the following lemma about the value of average-time games on timed
automata. 

\begin{lemma}
  \label{lemma:atg-regionally-constant}
  The value of every average-time game is regionally constant. 
\end{lemma}

An interesting implication of Lemma~\ref{lemma:atg-regionally-constant}
is that corner-point abstraction is sufficient to solve average-time
games with an arbitrary initial state.

\subsection{Determinacy of Average-Time Games on the Boundary Region
  Graph} 

Positional determinacy of average-time games on the boundary region
graph is immediate from Proposition~\ref{proposition:brg-is-finite}
and Theorem~\ref{theorem:apg-determined}.

\begin{theorem}
  \label{lemma:boundary-region-graph-determined}
  The average-time game on $\hTt$ is determined, and there are optimal 
  positional strategies in~$\hTt$, i.e., for every $q \in \cQ$, we
  have: 
  \[
  \VAL^{\hTt}(q) 
  = 
  \inf_{\mu \in \widehat{\Pi}_\mMIN} \sup_{\chi \in \hSigmaMax} 
  \AVERAGETIME_\mMIN(q, \mu, \chi)
  = 
  \sup_{\chi \in \widehat{\Pi}_\mMAX} \inf_{\mu \in \hSigmaMin} 
  \AVERAGETIME_\mMAX(q, \mu, \chi). 
  \]
\end{theorem}

In fact, in a boundary region graph, there are optimal type-preserving
boundary strategies. 
Before we show that, we need the following result.
\begin{lemma}
  \label{lemma:xi-best-response-in-hat}
  In $\hTt$, if $\mu \in \hSigmaMin$ and $\chi \in \hSigmaMax$ are
  mutual best responses from $q \in \cQ$, then 
  $\tpq{\mu}{q} \in \simpleMin$ and $\tpq{\chi}{q} \in \simpleMax$ are
  mutual best responses from every $q' \in \cQ([q])$. 
\end{lemma}

\begin{proof}
  We argue that $\tpq{\chi}{q}$ is a best response to $\tpq{\mu}{q}$ from 
  $q' \in \cQ([q])$ in $\hTt$;
  the other case is analogous.
  For all $\cchi \in \hSigmaMax$, we have the following: 
  \begin{multline*}
    \AVERAGETIME_\mMIN(q', \tpq{\mu}{q}, \tpq{\chi}{q})
    = 
    \AVERAGETIME_\mMIN(q, \tpq{\mu}{q}, \tpq{\chi}{q})
    \geq 
    \AVERAGETIME_\mMIN(q, \tpq{\mu}{q}, \tpq{\cchi}{q'})
    = 
    \\
    \AVERAGETIME_\mMIN(q', \tpq{\mu}{q}, \tpq{\cchi}{q'})
    = 
    \AVERAGETIME_\mMIN(q', \tpq{\mu}{q}, \cchi).
  \end{multline*}
  The first equality follows from
  Proposition~\ref{prop:rcbs-are-regionally-constant};
  the inequality follows because~$\chi$ is a best response to~$\mu$
  from~$q$;
  the second equality follows from
  Proposition~\ref{prop:rcbs-are-regionally-constant} again;
  and the last equality is straightforward.
\qed
\end{proof}

\begin{theorem}
  \label{theorem:optimal-xi-in-hat}
  There are optimal type-preserving boundary strategies in $\hTt$,
  i.e., for every $q \in \cQ$, we have:
  \[
  \VAL^{\hTt}(q) 
  = 
  \inf_{\mu \in \simpleMin} \sup_{\chi \in \hSigmaMax} 
  \AVERAGETIME_\mMIN(q, \mu, \chi)
  = 
  \sup_{\chi \in \simpleMax} \inf_{\mu \in \hSigmaMin} 
  \AVERAGETIME_\mMAX(q, \mu, \chi). 
  \]
\end{theorem}

\begin{proof}
  Let $\mu^* \in \simpleMin$ and $\chi^* \in \simpleMax$ be mutual
  best responses in $\hTt$;
  existence of such strategies follows from
  Lemma~\ref{lemma:xi-best-response-in-hat}. 
  Moreover, we can assume that the strategies $\mu^*$ and $\chi^*$
  have finite memory; 
  this can be achieved by taking positional strategies 
  $\mu \in \hSigmaMin$ and $\chi \in \hSigmaMax$ in 
  Lemma~\ref{lemma:xi-best-response-in-hat}.
  We then have the following:
  \begin{multline*}
    \inf_{\mu \in \simpleMin} \sup_{\chi \in \hSigmaMax} 
    \AVERAGETIME_\mMIN(q, \mu, \chi)
    \leq 
    \sup_{\chi \in \hSigmaMax} \AVERAGETIME_\mMIN(q, \mu^*, \chi)
    = 
    \AVERAGETIME_\mMIN(q, \mu^*, \chi^*)
    = 
    \\
    \AVERAGETIME_\mMAX(q, \mu^*, \chi^*)
    = 
    \inf_{\mu \in \hSigmaMin} \AVERAGETIME_\mMAX(q, \mu, \chi^*)
    \leq 
    \sup_{\chi \in \simpleMax} \inf_{\mu \in \hSigmaMin} 
    \AVERAGETIME_\mMAX(q, \mu, \chi).
  \end{multline*}
  The first and last inequalities are straightforward as
  $\mu^* \in \simpleMin$ and $\chi^* \in \simpleMax$.
  The first equality holds because $\chi^*$ is a best response
  to~$\mu^*$ in~$\hTt$, and the third equality holds because $\mu^*$
  is a best response to~$\chi^*$ in~$\hTt$. 
  Finally, the second equality holds because strategies~$\mu^*$
  and~$\chi^*$ have finite memory.
\qed
\end{proof}

\subsection{Determinacy of Average-Time Games on the Closed Region Graph}
To be able to show the determinacy of the average-time games on the
closed region graph, we need the following intermediate result.
\begin{lemma}
  \label{lemma:xi-best-response-in-bar}
  In $\cTt$, for every strategy in $\simpleMin$ there is a best
  response in $\simpleMax$, and for every strategy in $\simpleMax$
  there is a best response in $\simpleMin$.
\end{lemma}
\begin{proof}
  We argue that if $\mu \in \cSigmaMin$ is best-response to $\chi
  \in \simpleMax$ from $q \in \cQ$ then the strategy 
  $\tpqs{\mu}{q}{\chi}$ is best-response to $\chi$ from every 
  $q' \in \cQ([q])$.
  For all $\mmu \in \cSigmaMin$ we have the following:
  \begin{multline*}
    \AVERAGETIME_\mMIN(q', \tpqs{\mu}{q}{\chi}, \chi) = 
    \AVERAGETIME_\mMIN(q, \tpqs{\mu}{q}{\chi}, \chi) \leq 
    \AVERAGETIME_\mMIN(q, \mu, \chi) \leq   
    \AVERAGETIME_\mMIN(q, \tpqs{\mmu}{q'}{\chi}, \chi) = \\
    \AVERAGETIME_\mMIN(q', \tpqs{\mmu}{q'}{\chi}, \chi) \leq
    \AVERAGETIME_\mMIN(q', \mu', \chi).
  \end{multline*}
  The first and the second equalities follow from
  Proposition~\ref{prop:rcbs-are-regionally-constant}; 
  the second inequality follows because~$\mu$ is a best response
  to~$\chi$ from~$q$; and the first and the third inequalities
  follow from the the
  Corollary~\ref{corollary:type-preserving-is-better}.
  It follows that in $\cTt$ for every strategy $\chi \in \simpleMax$
  there is a best response in $\simpleMin$.
  Similarly we prove that in $\cTt$ for every strategy 
  $\mu \in \simpleMin$ there is a best response in $\simpleMax$.
\qed
\end{proof}

\begin{theorem}
  \label{theorem:optimal-xi-in-bar}
  The average-time game on $\cTt$ is determined, and there are optimal
  type-preserving boundary strategies in $\cTt$, 
  i.e., for every $q \in \cQ$, we have:
  \[
  \VAL^{\cTt}(q) 
  = 
  \inf_{\mu \in \simpleMin} \sup_{\chi \in \cSigmaMax} 
  \AVERAGETIME_\mMIN(q, \mu, \chi)
  = 
  \sup_{\chi \in \simpleMax} \inf_{\mu \in \cSigmaMin} 
  \AVERAGETIME_\mMAX(q, \mu, \chi)
  = 
  \VAL^{\hTt}(q).
  \]
\end{theorem}

\begin{proof}
  We have the following:
  \begin{multline*}
    \inf_{\mu \in \simpleMin} \sup_{\chi \in \cSigmaMax} 
    \AVERAGETIME_\mMIN(q, \mu, \chi)
    =  
    \inf_{\mu \in \simpleMin} \sup_{\chi \in \simpleMax} 
    \AVERAGETIME_\mMIN(q, \mu, \chi)
    = 
    \\
    \sup_{\chi \in \simpleMax} \inf_{\mu \in \simpleMin} 
    \AVERAGETIME_\mMAX(q, \mu, \chi)
    = 
    \sup_{\chi \in \simpleMax} \inf_{\mu \in \cSigmaMin} 
    \AVERAGETIME_\mMAX(q, \mu, \chi),
  \end{multline*}
  where the first and last equalities follow from
  Lemma~\ref{lemma:xi-best-response-in-bar}, and the second equality
  follows from Theorem~\ref{theorem:optimal-xi-in-hat}. 
  
  Now we show that $\LVAL^{\cTt}(q) \geq \LVAL^{\hTt}(q)$. 
  The proof that  $\UVAL^{\cTt}(q) \leq \UVAL^{\hTt}(q)$ is similar
  and hence omitted.
  \begin{multline*}
    \LVAL^{\cTt}(q) = \sup_{\chi \in \cSigmaMax} \inf_{\mu \in \cSigmaMin}
    \AVERAGETIME_\mMAX(q, \mu, \chi)
    \geq \sup_{\chi \in \simpleMax} \inf_{\mu \in \cSigmaMin}
    \AVERAGETIME_\mMAX(q, \mu, \chi) \\
    = \sup_{\chi \in \simpleMax} \inf_{\mu \in \simpleMin}
    \AVERAGETIME_\mMAX(q, \mu, \chi) = \LVAL^{\hTt}(q).
  \end{multline*}
  The first inequality follows as $\simpleMax \subseteq \cSigmaMax$. 
  The first equality holds by definition, the second
  equality is proved in the first paragraph of this proof, and the
  third equality follows from
  Theorem~\ref{theorem:optimal-xi-in-hat}. 
  From Lemma~\ref{lemma:boundary-region-graph-determined} we know that
  $\LVAL^{\hTt}(q) = \UVAL^{\hTt}(q)$.
  It follows that the average-time game on $\cTt$ is determined, and
  there are optimal type-preserving boundary strategies in $\cTt$. 
\qed
\end{proof}

\subsection{Determinacy of Average-Time Games on the Region Graph}

\begin{lemma}
\label{lemma:opt-strat-is-epsilon-opt}
If the strategies $\mu^* \in \simpleMin$ and $\chi^* \in \simpleMax$
are optimal for respective players in $\cTt$ then 
for every $\varepsilon > 0$, we have that 
\begin{eqnarray*}
\sup_{\chi \in \cSigmaMax} \AVERAGETIME_\mMIN(q, \mu^*_\varepsilon, \chi) 
  \leq  \VAL^{\cTt}(q) + \varepsilon~\text{ and } 
\inf_{\mu \in \cSigmaMin} \AVERAGETIME_\mMAX(q, \mu, \chi^*_\varepsilon) 
\geq  \VAL^{\cTt}(q) - \varepsilon,
\end{eqnarray*}
for all $\mu^*_\varepsilon \in \EcMin{\mu^*}{\varepsilon}$ and
$\chi^*_\varepsilon \in \EcMax{\chi^*}{\varepsilon}$.
\end{lemma}
\begin{proof}
  Let $\mu^* \in \simpleMin$ and $\chi^* \in \simpleMax$
  are optimal for respective players in $\cTt$.
  For all $\chi \in \cSigmaMax$, $\varepsilon > 0$, 
  and $\mu^*_\varepsilon \in \EcMin{\mu^*}{\varepsilon}$,
  we have the following:
  \begin{multline*}
    \AVERAGETIME_\mMIN(q, \mu^*_\varepsilon, \chi)
    \leq 
    \AVERAGETIME_\mMIN(q, \mu^*, \tpqs{\chi}{q}{\mu^*_\varepsilon}) + \varepsilon
    \leq 
    \AVERAGETIME_\mMIN(q, \mu^*, \chi^*) + \varepsilon
    = 
    \VAL^{\cTt}(q) + \varepsilon.
  \end{multline*}
  The first inequality is by
  Corollary~\ref{corollary:boundary-better-against-epsilon}. 
  The second inequality holds because $\chi^*$ is an optimal strategy
  and the equality is due to the fact that $\mu^*$ and $\chi^*$ are
  optimal.
\qed 
\end{proof}

\begin{theorem}
  \label{theorem:region-graph-is-determined}
  The average-time game on $\rTt$ is determined, and for every $q \in
  \cQ$, we have $\VAL^{\rTt}(q) = \VAL^{\cTt}(q)$. 
\end{theorem}
\begin{proof}
Let $\mu^* \in \simpleMin$ be an optimal strategy of player Min in $\cTt$.
Let us fix an $\varepsilon > 0$ and $\mu^*_\varepsilon \in
\EcMin{\mu^*}{\varepsilon}$.  
\begin{multline*}
  \UVAL^{\rTt}(q) =  \inf_{\mu \in \rSigmaMin} \sup_{\chi \in \rSigmaMax} 
  \AVERAGETIME_\mMIN(q, \mu, \chi)
  \leq  \sup_{\chi \in \rSigmaMax} \AVERAGETIME_\mMIN(q, \mu^*_\varepsilon, \chi)
  \leq \\
  \sup_{\chi \in \cSigmaMax} \AVERAGETIME_\mMIN(q, \mu^*_\varepsilon, \chi)
  \leq \VAL^{\cTt}(q) + \varepsilon.
\end{multline*}
The second inequality follows because 
${\mu^*_\varepsilon \in \rSigmaMin}$ and the third inequality follows
 as ${\rSigmaMax \subseteq \cSigmaMax}$.
The last inequality follows from
Lemma~\ref{lemma:opt-strat-is-epsilon-opt} because $\mu^* \in
\simpleMin$ is an optimal strategy in $\cTt$. 
Similarly we show that for every $\varepsilon > 0$ we have that 
$\LVAL^{\rTt}(q) \geq \VAL^{\cTt}(q) - \varepsilon$.
Hence it follows that $\VAL^{\rTt}(q)$ exists and its value 
is equal to $\VAL^{\cTt}(q)$. 
\qed
\end{proof}

\section{Complexity}
\label{section:complexity}

The main decision problem for average-time game is as follows:
given an average-time game  
$\Gamma=(\Tt, L_\mMIN, L_\mMAX)$, a state $s \in S$, and a number
$B \in \Rplus$, decide whether $\VAL(s) \leq B$.

\begin{theorem}
  \label{theorem:average-time-games-are-expc}
  Average-time games are EXPTIME-complete on timed automata with at
  least two clocks. 
\end{theorem}
\begin{proof}
  From Theorem~\ref{theorem:atg-on-graphs-are-determined}
  we know that in order to solve an average-time game starting from an
  initial state of a timed automaton, it is sufficient to solve the
  average-time game on the set of states of the boundary region graph of
  the automaton that are reachable from the initial state.
  Observe that every region, and hence also every configuration of
  the game, can be represented in space polynomial in the size of the
  encoding of the timed automaton and of the encoding of the initial
  state, and that every move of the game can be simulated in polynomial
  time.  
  Therefore, the value of the game can be computed by a straightforward
  alternating PSPACE algorithm, and hence the problem is in EXPTIME
  because APSPACE $=$ EXPTIME.
  
  In order to prove EXPTIME-hardness of solving average-time games on
  timed automata with two clocks, we reduce the EXPTIME-complete
  problem of solving countdown games~\cite{JLS07} to it. 
  Let $G = (N, M, \pi, n_0, B_0)$ be a countdown game, where $N$ is a 
  finite set of nodes, $M \subseteq N \times N$ is a set of moves, 
  $\pi : M \to \Npos$ assigns a~positive integer number to every move,
  and $(n_0, B_0) \in N \times \Npos$ is the initial configuration. 
  
  W.l.o.g we assume that there is an integer $W$ such that 
  $\pi(n_1, n_2) \geq W$ for every move $(n_1, n_2) \in M$.
  $(n, B) \in N \times \Npos$, first player~1 chooses a number 
  $p \in \Npos$, such that $p \leq B$ and $\pi(n, n') = p$ for some
  move $(n, n') \in M$, and then player~2 chooses a move 
  $(n, n'') \in M$, such that $\pi(n, n'') = p$;
  the new configuration is then $(n'', B - p)$. 
  Player~1 wins a play of the game when a configuration $(n, 0)$ is
  reached, and he loses (i.e., player~2 wins) when a configuration
  $(n, B)$ is reached in which player~1 is stuck, i.e., for all moves
  $(n, n') \in M$, we have $\pi(n, n') > B$.
  
  We define the timed automaton 
  $\Tt_G = (L, C, S, A, E, \delta, \xi, F)$ by setting 
  $C = \eset{b, c}$; $S = L \times (\REALS{B_0})^2$; 
  $A = \eset{*} \cup P \cup M$, where $P = \pi(M)$, the image of the
  function $\pi : M \to \Npos$;  
  \begin{eqnarray*}
    L & = & \eset{*} \cup N \cup 
    \Set{(n, p) \: : \: \text{ there is } (n, n') \in M,
      \text{s.t. } \pi(n, n') = p};
    \\
    E(a) & = & 
    \begin{cases}
      \set{(n, \nu) \: : \: n \in N \text{ and } \nu(b) = B_0}
      ~\text{if $a = *$},
      \\
      \set{(*, \nu) \: : \: \nu(c) = W}
      ~\text{if $a = *$},
      \\
      \Set{(n, \nu) \: : \: 
        \exists (n, n') \in M, \text{s.t. } \pi(n, n') = p 
        \text{ and } \nu(c) = 0}
      ~\text{if $a = p \in P$},
      \\
      \Set{\big((n, p), \nu\big) \: : \:
        \pi(n, n') = p \text{ and } \nu(c) = p} 
      ~\text{if $a = (n, n') \in M$},
    \end{cases}
    \\
    \delta(\ell, a) & = & 
    \begin{cases}
      * & \text{if $\ell = n \in N$ and $a = *$},
      \\
      * & \text{if $\ell = *$ and $a = *$},
      \\
      (n, p) & \text{if $\ell = n \in N$ and $a = p \in P$}, 
      \\
      n' & \text{if $\ell = (n, p) \in N \times P$ and 
        $a = (n, n') \in M$}; 
    \end{cases}
  \end{eqnarray*}
  $\xi(a) = \eset{c}$, for every $a \in A \setminus \eset{*}$ and
  $\xi(*) = \eset{b, c}$.
  Note that the timed automaton $\Tt_G$ has only two clocks and that
  the clock $b$ is reset only in the special location $*$.
  
  Finally, we define the average-time game on timed game automaton
  $\Gamma_G = (\Tt_G, L_1, L_2)$ by setting $L_1 = N$ and 
  $L_2 = L \setminus L_1$.
  It is routine to verify that value of the average-time game at the 
  state $(n_0, (0, 0)) \in S$ is $W$ in the average-time game
  on~$\Gamma_G$ if and only if player~1 has a winning strategy 
  (from the initial configuration $(n_0, B_0)$)
  in the countdown game~$G$.  
  \qed
\end{proof}
\vskip8pt
\noindent {\bf Acknowledgments.}
This work was partially supported by the EPSRC grants EP/E022030/1 and
EP/F001096/1. 

\bibliographystyle{plain} 
\bibliography{papers}

\newpage
\appendix
\section{Proof of Proposition~\ref{proposition:time-is-regionally-simple}}
In order to prove this proposition, we need the following result.

\begin{proposition}[\cite{JT07,Tri09}]
  \label{proposition:atg-t-alpha-simple}
  Let $\alpha \in \aA$ and regions $R, R', R'' \in \Rr$ be such that 
  $R \xrightarrow{R''}_\alpha R'$.
  If $F : \pQ(R') \to \Real$ is simple
  then $F^\oplus_{(\alpha, R'')}: \pQ(R) \to \Real$,
  defined as $(s, R) \mapsto t(s, \alpha) + F(\SUCC(q, (\alpha, R'')))$,
  is simple.   
\end{proposition}

\begin{proof}[Proof of Proposition~\ref{proposition:time-is-regionally-simple}]
  Let $\mu \in \simpleMin$ and $\chi \in \simpleMax$.
  We prove this lemma by induction on the value of $n$. 
  The base case for $n = 0$ is trivial. 
  Assume that for every $\mu \in \simpleMin$ and 
  $\chi \in \simpleMax$ the function 
  $\TIME(\RUN_{k}(\cdot, \mu, \chi)): \cQ \to \Rplus$ is regionally  
  simple. 
  To prove this proposition we now need to show that for 
  $\mu \in \simpleMin$ and $\chi \in \simpleMax$ the function 
  $\TIME(\RUN_{k+1}(\cdot, \mu, \chi))$ is
  regionally simple.
  
  Let the strategies $\mu' \in \simpleMin$ and $\chi' \in
  \simpleMax$ be such that for every $q \in \cQ$ the run
  $\RUN_{k}(\SUCC(q, \mu, \chi), \mu', \chi')$ be the length $k$
  suffix of $\RUN_{k+1}(q, \mu, \chi)$.
  From inductive hypothesis we have that 
  $\RUN_{k}(\cdot, \mu', \chi')$ is regionally simple.
  Assume that $R \in \Rr_\mMIN$ and let 
  $\BS{\mu}(\seq{q}) = (\alpha, R'')$ for
  every $q \in \cQ(R)$.
  The treatment for the case where $R \in \Rr_\mMAX$ is similar.
  Now for every $q = (s, R) \in \cQ(R)$ we have that 
  $\TIME(\RUN_{k+1}(q, \mu, \chi)) =
  t(s, \alpha) + \RUN_{k}(\SUCC(q, (\alpha, R'')), \mu', \chi')$, 
  which from Proposition~\ref{proposition:atg-t-alpha-simple} is a
  simple function.
  \qed
\end{proof}

\section{Proof of Proposition~\ref{prop:rcbs-are-regionally-constant}}
\begin{proof}
  Let $\mu \in \simpleMin$, $\chi \in \simpleMax$ 
  and $q = (s, R), q' = (s', R) \in \cQ(R)$. 
  We have
  \begin{multline*}
    \AVERAGETIME_\mMIN(q, \mu, \chi) - \AVERAGETIME_\mMIN(q', \mu,
    \chi) \\
    = \liminf_{n \to \infty} (1/n) \cdot 
    \TIME(\RUN_n(q, \mu,\chi)) - \liminf_{n \to \infty}(1/n)\cdot
    \TIME(\RUN_n(q', \mu, \chi)) \\ 
    = \liminf_{n \to \infty} (1/n) \cdot 
    \left(b - s(c) - b + s'(c)\right) 
    = \liminf_{n \to \infty} (1/n) \cdot \left(s'(c) - s(c)\right) = 0.
  \end{multline*}
  The first equality is by definition, the second follows from
  Proposition~\ref{proposition:time-is-regionally-simple}, and the
  last two equalities are trivial.
  In a similar manner we show that $\AVERAGETIME_\mMAX(q, \mu, \chi) =
  \AVERAGETIME_\mMAX(q', \mu, \chi)$. 
\qed
\end{proof}

\section{Proof of Proposition~\ref{proposition:nstep-type-preserving-is-better}}
In order to prove this
Proposition~\ref{proposition:nstep-type-preserving-is-better} and
Proposition~\ref{proposition:nstep-boundary-better-against-epsilon}, 
we need the following result.
\begin{proposition}[\cite{JT07,Tri09}]
  \label{proposition:atg-nondecreasing}
  Let $a \in A$ and regions $R, R', R'' \in \Rr$ be such that 
  $R \xrightarrow{}_{*} R'' \xrightarrow{a} R'$.
  If $F : \pQ(R') \to \Real$ is simple then for every 
  $q = (s, R) \in \pQ(R)$, function 
  $F^\oplus_{(q, R'', a)}: I \to \Real$, defined as 
  $t \mapsto  t + F(\SUCC(q, (t, R'', a)))$, is continuous and
  nondecreasing, where    
  $I = \set{t \in \Rplus \: : \: (s + t) \in \CLOSOP(R'')}$.  
\end{proposition}

\begin{proof}[Proof of
  Proposition~\ref{proposition:nstep-type-preserving-is-better}] 
  The proof is by induction on $n$.
  The base case, when $n = 0$, is trivial. 
  In the rest of the proof we show that for 
  $\chi \in \simpleMax$, $\mu \in \cSigmaMin$, and
  a configuration $q = (s, R) \in \cQ$, we have that 
  $\TIME(\RUN_{k+1}(q, \mu, \chi)) \geq 
  \TIME(\RUN_{k+1}(q, \tpqs{\mu}{q}{\chi}, \chi))$ 
  assuming that the proposition holds for $n = k$.
  The proof for the case where $q \in \cQ_\mMAX$ is trivial. 
  In the rest of the proof we assume that 
  $q \in \cQ_\mMIN$.
  
  Let us fix $\chi \in \simpleMax$ and 
  $\mu \in \cSigmaMin$.
  Let $\RUN_{k+1}(q, \mu, \chi)$ and 
  $\RUN_{k+1}(q, \tpqs{\mu}{q}{\chi}, \chi)$  be 
  $\seq{q_0, \tau_1, q_1, \ldots, q_{k+1}}$ 
  and $\seq{q_0', \tau_1', q_1', \ldots, q_{k+1}'}$, respectively,
  where $q_0 = q_0' = q$. 
  Notice that by definition the run types of both runs are the same. 
  Hence for every index $i \leq k+1$
  we have $q_i = (s_i, R_i)$ and $q_i' = (s_i', R_i)$, and  
  for every index $i \leq k+1$ we have 
  $\tau_i = (t_i, R_i', a_i)$ and $\tau_i' = (t_i', R_i', a_i)$.
  
  Let $\cchi \in \simpleMax$ and $\mmu \in \cSigmaMin$ be such 
  that $\RUN_{k}(q_1, \mmu, \cchi)$ 
  be length $k$ suffix of $\RUN_{k+1}(q, \mu, \chi)$.
  Notice that we assume that $\cchi$ is type-preserving.
  It is easy to see that 
  \[
  \TIME(\RUN_{k+1}(q, \mu, \chi)) 
  = t_1 + \TIME(\RUN_k(q_1, \mmu, \cchi)).
  \] 
  From inductive hypothesis, we get that 
  \begin{equation}
    \TIME(\RUN_{k+1}(q, \mu, \chi)) 
    \geq 
    t_1 + \TIME(\RUN_k(q_1, \tpqs{\mmu}{q_1}{\cchi},
    \cchi)).~\label{atg-ashu1} 
  \end{equation}
  Since the strategies $\tpqs{\mmu}{q_1}{\cchi} \in \simpleMin$ and
  $\cchi \in \simpleMax$ are type-preserving, from 
  Proposition~\ref{proposition:time-is-regionally-simple} we get that 
  $\TIME(\RUN_k(\cdot, \tpqs{\mmu}{q_1}{\cchi}, \cchi))$ is
  regionally simple. 
  Let us denote the restriction of this function on domain 
  $\cQ(R_1)$ by $\Ff: \cQ(R_1) \to \Real$. 
  Let us define the partial function  
  $\Ff^\oplus_{(q, R_1', a)} : \Rplus \rightharpoondown \Real$ as
  $t \mapsto  t + \Ff(\SUCC(q, (t, R'', a)))$, 
  for all $t \in \Rplus$, such that $(s+t) \in \CLOSOP(R_1')$.
  The following inequality follows from (\ref{atg-ashu1}): 
  \[
  \TIME(\RUN_{k+1}(q, \mu, \chi)) 
  \geq t_1 + \Ff(q_1) \geq 
  \inf_{t}  \Set{\Ff^\oplus_{(q, R_1', a)}(t) 
    \::\: s + t \in \CLOSOP(R_1')}. 
  \]
  Since $\tpqs{\mu}{q}{\chi}$ is a type-preserving boundary strategy
  of player Min, from equation (\ref{e:brg-strat-inf}), we know that
  $t_1' = \inf \set{ t \::\: s + t \in \CLOSOP(R_1')}$.
  Moreover from Proposition~\ref{proposition:atg-nondecreasing} we
  have that $\Ff^\oplus_{(q, R_1', a)}$ is continuous and
  nondecreasing on the domain 
  ${\set{t \in \Rplus \: : \: (s + t) \in \CLOSOP(R'')}}$.
  Hence $\Ff^\oplus_{(q, R_1', a)}(t_1') 
  = \inf_{t}  \Set{\Ff^\oplus_{(q, R_1', a)}(t) 
    \::\: s + t \in \CLOSOP(R_1')}$. 
  Combining these facts, we get the following inequalities: 
  \[
  \TIME(\RUN_{k+1}(q, \mu, \chi)) 
  \geq  \Ff^\oplus_{(q, R_1', a)}(t_1')
  =  t_1' + \TIME(\RUN_k(q_1', \tpqs{\mmu}{q_1}{\cchi}, \cchi)) 
  \]
  
  Since  
  $\RUN_k(q_1', \tpqs{\mmu}{q_1}{\cchi}, \cchi)$
  is length $k$ suffix of  
  $\RUN_{k+1}(q, \tpqs{\mu}{q}{\chi}, \chi)$, we get the desired
  inequality. 
\qed
\end{proof}

\section{Proof of
  Proposition~\ref{proposition:nstep-boundary-better-against-epsilon}}
\begin{proof}
  The proof is by induction on $n$.
  The base case, when $n = 0$, is trivial. 
  In the rest of the proof we show that for 
  $\chi \in \simpleMax$, $\mu \in \cSigmaMin$, $\varepsilon > 0$, 
  $\chi_\varepsilon \in \EcMax{\chi}{\varepsilon}$, 
  and a configuration $q = (s, R) \in \cQ$, we have that 
  $\TIME(\RUN_{k+1}(q, \mu, \chi_\varepsilon)) \geq 
  \TIME(\RUN_{k+1}(q, \tpqs{\mu}{q}{\chi_\varepsilon}, \chi)) -k\cdot
  \varepsilon$,   
  assuming that the proposition holds for $n = k$.
  
  Let us fix $\chi \in \simpleMax$,  
  $\mu \in \cSigmaMin$, $\varepsilon > 0$, and 
  $\chi_\varepsilon \in \EcMax{\chi}{\varepsilon}$.
  Let $\RUN_{k+1}(q, \mu, \chi_\varepsilon)$ and 
  $\RUN_{k+1}(q, \tpqs{\mu}{q}{\chi_\varepsilon}, \chi)$  be 
  $\seq{q_0, \tau_1, q_1, \ldots, q_{k+1}}$ 
  and $\seq{q_0', \tau_1', q_1', \ldots, q_{k+1}'}$, respectively,
  where $q_0 = q_0' = q$. 
  Notice that by definition the run types of both runs are the same. 
  Hence for every index $i \leq k+1$
  we have $q_i = (s_i, R_i)$ and $q_i' = (s_i', R_i)$, and  
  for every index $i \leq k+1$ we have 
  $\tau_i = (t_i, R_i', a_i)$ and $\tau_i' = (t_i', R_i', a_i)$.
  
  Let $\cchi \in \simpleMax$ and $\mmu \in \cSigmaMin$ be such 
  that $\RUN_{k}(q_1, \mmu, \cchi_\varepsilon)$ 
  be length $k$ suffix of  
  $\RUN_{k+1}(q, \mu, \chi_\varepsilon)$.
  Notice that we assume that $\cchi$ is type-preserving.
  It is easy to see that 
  \[
  \TIME(\RUN_{k+1}(q, \mu, \chi_\varepsilon)) 
  = t_1 + \TIME(\RUN_k(q_1, \mmu, \cchi_\varepsilon)).
  \] 
  From inductive hypothesis, we get that 
  \begin{equation}
    \TIME(\RUN_{k+1}(q, \mu, \chi_\varepsilon)) 
    \geq 
    t_1 + \TIME(\RUN_k(q_1, \tpqs{\mmu}{q_1}{\cchi_\varepsilon},
    \cchi)) -k \cdot \varepsilon.~\label{atg-ashu-proof2} 
  \end{equation}
  Since the strategies $\tpqs{\mmu}{q_1}{\cchi_\varepsilon} \in
  \simpleMin$ and 
  $\cchi \in \simpleMax$ are type-preserving boundary strategies, from
  Proposition~\ref{proposition:time-is-regionally-simple} we get that 
  $\TIME(\RUN_k(\cdot, \tpqs{\mmu}{q_1}{\cchi_\varepsilon}, \cchi))$ is
  regionally simple. 
  Let us denote the restriction of this function on domain 
  $\cQ(R_1)$ by $\Ff: \cQ(R_1) \to \Real$. 
  Let us define the partial function  
  $\Ff^\oplus_{(q, R_1', a)} : \Rplus \rightharpoondown \Real$ as
  $t \mapsto  t + \Ff(\SUCC(q, (t, R'', a)))$, 
  for all $t \in \Rplus$, such that $(s+t) \in \CLOSOP(R_1')$.
  The following inequality follows from (\ref{atg-ashu-proof2}): 
  \[
  \TIME(\RUN_{k+1}(q, \mu, \chi_\varepsilon)) 
  \geq t_1 + \Ff(q_1) - k \cdot \varepsilon \geq 
  \inf_{t}  \Set{\Ff^\oplus_{(q, R_1', a)}(t) 
    \::\: s + t \in \CLOSOP(R_1')} - k\cdot \varepsilon. 
  \]
  We need to consider two cases:~$q \in \cQ_\mMIN$ and $q \in
  \cQ_\mMAX$. 
  \begin{itemize}
  \item 
    Assume that $q \in \cQ_\mMIN$.
    Since $\tpqs{\mu}{q}{\chi_\varepsilon}$ is a type-preserving
    boundary strategy 
    of player Min, from equation (\ref{e:brg-strat-inf}), we know that
    $t_1' = \inf \set{ t \::\: {s + t \in \CLOSOP(R_1')}}$.
    Moreover from Proposition~\ref{proposition:atg-nondecreasing} we
    have that $\Ff^\oplus_{(q, R_1', a)}$ is continuous and
    nondecreasing on the domain 
    $\set{t \in \Rplus \: : \: {(s + t) \in \CLOSOP(R'')}}$.
    Hence $\Ff^\oplus_{(q, R_1', a)}(t_1') 
    = \inf_{t}  \Set{\Ff^\oplus_{(q, R_1', a)}(t) 
      \::\: s + t \in \CLOSOP(R_1')}$. 
    Combining these facts, we get the following inequalities: 
    \[
    \TIME(\RUN_{k+1}(q, \mu, \chi_\varepsilon)) 
    \geq  t_1' + \TIME(\RUN_k(q_1',
    \tpqs{\mmu}{q_1}{\cchi_\varepsilon}, \cchi)) - k \cdot \varepsilon. 
    \]
    Since  
    $\RUN_k(q_1', \tpqs{\mmu}{q_1}{\cchi_\varepsilon}, \cchi)$
    is length $k$ suffix of  
    $\RUN_{k+1}(q, \tpqs{\mu}{q}{\chi_\varepsilon}, \chi)$, we get the
    following inequality:
    \begin{eqnarray*}
      \TIME(\RUN_{k+1}(q, \mu, \chi_\varepsilon)) 
      &\geq& \TIME(\RUN_{k+1}(q, \tpqs{\mu}{q_1}{\chi_\varepsilon}, \chi))
      - k \cdot \varepsilon\\
      &\geq & \TIME(\RUN_{k+1}(q, \tpqs{\mu}{q_1}{\chi_\varepsilon}, \chi))
      - (k+1) \cdot \varepsilon,
    \end{eqnarray*}
    as required.
  \item 
    Assume that $q \in \cQ_\mMAX$.
    So far we have shown that 
    \begin{equation}
      \TIME(\RUN_{k+1}(q, \mu, \chi_\varepsilon)) 
      \geq t_1 + \Ff(q_1) - k \cdot \varepsilon.
      \label{eqn:ashu-another-interesting-eqn}
    \end{equation}
    Since $\Ff$ is a simple function let
    $\Ff((s_1, R_1)) =  b - s_1(c)$ for all 
    $(s_1, R_1) \in \cQ(R_1)$. 
    For all $t \in \Rplus$ such that $s + t \in R_1'$ we have the
    following observation.
    \begin{equation}
    t + \Ff((\SUCC(s, (t, a_1))) = 
    \begin{cases}
      t + b &~\text{ if $c \in \xi(a_1)$ }\\
      b - s(c)&~\text{ otherwise.}
    \end{cases}
    \label{eqn:ashu-atg-simple-eqn}
    \end{equation}
    By Definition~\ref{def:epsilon-close-strat} we know that $t_1 \geq
    t_1' - \varepsilon$. 
    Combining this with (\ref{eqn:ashu-atg-simple-eqn}) we get that 
    \[
    t_1 + \Ff(q_1) \geq t_1' + \Ff(q_1') - \varepsilon.
    \]
    We can then rewrite (\ref{eqn:ashu-another-interesting-eqn}) as
    the following:
    \[
    \TIME(\RUN_{k+1}(q, \mu, \chi_\varepsilon)) 
    \geq t_1' + \Ff(q_1') - (k+1)\cdot \varepsilon.
    \]
    The term $\Ff(q_1')$ represents the sum of the times of
    $\RUN_k(q_1', \tpqs{\mmu}{q_1}{\cchi_\varepsilon}, \cchi)$.
    Since  
    $\RUN_k(q_1', \tpqs{\mmu}{q_1}{\cchi_\varepsilon}, \cchi)$
    is length $k$ suffix of
    $\RUN_{k+1}(q, \tpqs{\mu}{q}{\chi_\varepsilon}, \chi)$, we get the
     inequality
    \[
    \TIME(\RUN_{k+1}(q, \mu, \chi_\varepsilon)) 
    \geq \TIME(\RUN_{k+1}(q, \tpqs{\mu}{q}{\chi_\varepsilon}, \chi)) 
    - (k+1)\cdot \varepsilon,
    \]
    as required.
  \end{itemize}
\qed
\end{proof}

\end{document}